\input ./style/arxiv-general.cfg
\documentclass[aop,MSNbibl,secthm,seceqn,dvips]{arximspdf}
\makeatletter
   \@ifpackageloaded{graphicx}{}{\usepackage{graphicx}}
\makeatother

\doi{10.1214/15-AOP1039}
\volume{44}
\issue{4}
\pubyear{2016}
\firstpage{2980}
\lastpage{3031}
\docsubty{FLA}

\makeatletter
\newtheorem{theorem}{Theorem}[section]

\newtheorem{proposition}[theorem]{Proposition}
\newproclaim{conjecture}[theorem]{Conjecture}
\newtheorem{lemma}[theorem]{Lemma}
\newproclaim{Def}{Definition}
\newproclaim{definition}[theorem]{Definition}
\newproclaim{remark}[theorem]{Remark}
\makeatother

\begin{document}
\begin{frontmatter}

\title{Fractional Brownian motion with Hurst index $H=0$ and the
Gaussian Unitary Ensemble}
\runtitle{fBm with $H=0$ and the GUE}

\begin{aug}
\author[A]{\fnms{Y. V.}~\snm{Fyodorov}\thanksref{T1}\ead
[label=e1]{y.fyodorov@qmul.ac.uk}},
\author[A]{\fnms{B. A.}~\snm{Khoruzhenko}\ead[label=e2]{b.khoruzhenko@qmul.ac.uk}}
\and
\author[A]{\fnms{N. J.}~\snm{Simm}\corref{}\thanksref{T1}\ead
[label=e3]{n.simm@qmul.ac.uk}}
\runauthor{Y. V. Fyodorov, B. A. Khoruzhenko and N. J. Simm}
\affiliation{Queen Mary University of London}
\address[A]{School of Mathematical Sciences\\
Queen Mary University of London\\
Mile End Road\\
London E1 4NS\\
United Kingdom\\
\printead{e1}\\
\phantom{E-mail:\ }\printead*{e2}\\
\phantom{E-mail:\ }\printead*{e3}}
\end{aug}
\thankstext{T1}{Supported by the EPSRC Grant EP/J002763/1 ``Insights
into Disordered Landscapes via Random Matrix Theory and Statistical
Mechanics.''}

%
\received{\smonth{12} \syear{2013}}
%
\revised{\smonth{5} \syear{2015}}

%
\begin{abstract}
The goal of this paper is to establish a relation between
characteristic polynomials of $N \times N$ GUE random matrices
$\mathcal{H}$ as $N \to\infty$, and Gaussian processes with
logarithmic correlations.
We introduce a regularized version of fractional Brownian motion with
zero Hurst index, which is a Gaussian process with stationary
increments and logarithmic increment structure. Then we prove that this
process appears as a limit of $D_{N}(z)=-\log|\det(\mathcal{H}-zI)|$
on \textit{mesoscopic scales} as $N \to\infty$. By employing a
Fourier integral representation, we use this to prove a continuous
analogue of a result by Diaconis and Shahshahani
[\textit{J. Appl. Probab.} \textbf{31A} (1994) 49--62].
On the
\textit{macroscopic scale}, $D_{N}(x)$ gives rise to yet another type
of Gaussian process with logarithmic correlations. We give an explicit
construction of the latter in terms of a Chebyshev--Fourier random series.
\end{abstract}

%
\begin{keyword}[class=AMS]
\kwd[Primary ]{60B20}
\kwd[; secondary ]{60F17}
\kwd{60F05}
\kwd{15B52}
\end{keyword}
\begin{keyword}
\kwd{Random matrix theory}
\kwd{mesoscopic regime}
\kwd{logarithmically correlated}
\kwd{fractional Brownian motion}
\kwd{generalized processes}
\end{keyword}
\end{frontmatter}

\section{Introduction}\label{sec1}
Suppose that $\mathcal{H}$ is a random Hermitian matrix of size
$N\times N$ taken from the Gaussian Unitary Ensemble (GUE), with
ensemble distribution given by the measure 
%
%
\begin{equation}
\label{dens} 
\mathrm{Const.} \exp\bigl[-2N \operatorname{Tr} \bigl(
\mathcal{H}^{2} \bigr) \bigr] \prod_{j=1}^N
d\mathcal{H}_{jj} \prod_{1\le j<k\le N} d \operatorname{Re}
\mathcal{H}_{jk} d \operatorname{Im}\mathcal{H}_{jk}.
\end{equation}
It is well known that in the limit of infinite matrix dimensions $N\to
\infty$, the distribution of the eigenvalues of $\mathcal{H}$ is
supported on the interval $[-1,1]$ and has density $\frac{2}{\pi
}\sqrt{1-x^2}$ there. This is known as Wigner's semicircle law; see,
for example, \cite{PS11} and \cite{AGZ09} for precise statements. In
this paper, we are concerned with the \emph{random process} in $x$
defined by the logarithm
%
%
\begin{equation}
\label{krasov}
D_{N}(x) =-\log\bigl|\det( \mathcal{H} - xI)\bigr| 
\end{equation}
of the characteristic polynomial of $\mathcal{H}$ in the limit $N\to
\infty$, with $x$ varying in $(-1,1)$. The quantity\vspace*{1pt} $D_N(x)$ is a
particular case of linear eigenvalue statistics $X_{N}(f) = \sum
_{k=1}^{N} f(x_k)$, where $x_1, \dots, x_N$ are the eigenvalues of
$\mathcal{H}$. It is well known that for suitably regular test
functions\vspace*{1pt} $f$, $X_{N}(f)$ is asymptotically normal as $N \to\infty$
with variance $\sigma^{2}(f) = \frac{1}{4}\sum_{k=1}^{\infty
}kc_{k}(f)^{2}$, where $c_{k}(f)$ are the \textit{Chebyshev--Fourier
coefficients}:
%
\begin{equation}
\label{chebshev}
c_{k}(f) = \frac{2}{\pi}\int_{-1}^{1}
\frac{f(u)T_{k}(u)}{\sqrt
{1-u^{2}}}\,du, \qquad T_{k}(u) = \cos\bigl(k
\operatorname{arccos}(u) \bigr).
\end{equation}
In fact, the asymptotic normality of $X_{N}(f)$ for regular $f$ has
been established for a variety of random matrix ensembles; see, for
example, \cite{Joh98,LP09,PS11} and references therein.

Since $x$ lies in the bulk of the eigenvalue distribution, our test
function, $f(u)=\log|u-x|$ is unbounded. Its Chebyshev--Fourier
coefficients are proportional to $1/k$, so that $\sigma^{2}(f) =\infty
$ and it is then natural to consider normalizing $D_N(x)$ before taking
the limit $N \to\infty$. Indeed, for any fixed $x\in(-1,1)$ the
variance of $D_N(x)$ grows with $N$ like $\frac{1}{2}\log N$, and for
any finite number of distinct points $x_{1}, \ldots,x_{m} $ in
$(-1,1)$ the random vector $(D_{N}(x_{1}),\ldots,D_{N}(x_{m}))/(\frac
{1}{2} \log N)^{1/2}$ converges in distribution, after centering, to a
collection of $m$ independent standard Gaussians as $N\to\infty$.
This can be inferred from the asymptotic identity due to Krasovsky
\cite{K07}:
%
%
\begin{eqnarray}
\mathbb{E} \bigl\{e^{-\sum_{k=1}^{m}\alpha_{k}D_{N}(x_{k})}
\bigr\}
&= &\prod_{k=1}^{m} \biggl[C \biggl(
\frac{\alpha_{k}}{2} \biggr) \bigl(1-x_{k}^{2}
\bigr)^{\alpha_{k}^{2}/8}N^{\alpha_{k}^{2}/4}e^{\alpha
_{k}N(2x_{k}^{2}-1-2\log(2))/2} \biggr]
\nonumber
\\[-8pt]
\label{krasovasympt}
\\[-8pt]
\nonumber
&&{}\times\prod_{1 \leq\nu< \mu\leq m}\bigl(2|x_{\nu}-x_{\mu
}|\bigr)^{-\alpha
_{\nu}\alpha_{\mu}/2}
\biggl(1+O \biggl(\frac{\log N}{N} \biggr) \biggr),
\end{eqnarray}
where $C(\alpha) = 2^{2\alpha^{2}}{G(\alpha+1)^{2}}/{G(2\alpha+1)}$
and $G(z)$ is the Barnes G-function.
The most salient feature of the asymptotics in (\ref{krasovasympt}) is
the product of differences on the second line, which when rewritten in
the form
%
%
\begin{equation}\label{kracovstruct}
\exp\biggl[ -\sum_{1 \leq\nu< \mu\leq m}\frac{\alpha_{\nu
}\alpha
_{\mu}}{2}
\log\bigl|2(x_{\nu}-x_{\mu})\bigr| \biggr],
\end{equation}
is suggestive of the existence of a logarithmic covariance structure in
the Gaussian process $D_N(x)$. However, this term is of sub-leading
order to the variance term. Clearly then, the normalization of the
process (\ref{krasov}) comes at a price, because the nontrivial
covariance structure implied by (\ref{kracovstruct}) is too small to
survive the limit $N \to\infty$.

This motivates the following question. How can we ``regularize'' the
process~(\ref{krasov}) so that it has a well-defined limit that
``feels'' the covariance structure implied by (\ref{kracovstruct})?
Hughes, Keating and O'Connell \cite{HKOC01} answered this question in
the context of the Circular Unitary Ensemble (Haar unitary matrices).
Employing convergence in functional spaces instead of point-wise
convergence, they proved that the logarithm $V_N(\theta)=-2\log
{|p_N(\theta)|}$ of the characteristic polynomial $p_N(\theta)=\det
{ (I-U  e^{-i\theta} )}$ of Haar unitary matrices $U$
converges as $N\to\infty$ to the stochastic process represented by
the Fourier series
%
%
\begin{equation}
\label{1f}
V(\theta)=\sum_{n=1}^{\infty}
\frac{1}{\sqrt{n}} \bigl(v_n e^{i n
\theta}+\overline{v}_n
e^{-i n \theta} \bigr).
\end{equation}
Here, the coefficients $v_n,\overline{v}_n$ are independent standard
\emph{complex} Gaussians, $\mathbb{E} \{v_n \overline{v}_n \}=1$,
and the convergence of the series is understood in the sense of
distributions in a suitable Sobolev space. This\vspace*{1pt} process has a
logarithmic singularity in the covariance structure:
$\mathbb{E} \{V(\theta_1)V(\theta_2)\}=-2\log{|e^{i\theta
_1}-e^{i\theta_2}|}$.

At this point, it is appropriate to mention that random processes and
fields with logarithmic covariance structure appear with astonishing
regularity in physics and also engineering applications; see, for
example, \cite{CLD01} and more recently \cite{FlDR12}. Those objects
are intimately related to multi-fractal cascades emerging in
turbulence, and from that angle attracted considerable mathematical
interest within the last decade; see, for example,
\cite{BacryMuzy2002} and \cite{BarralMandelbrot2004}. In fact, closely related
mathematical objects appear in the so-called ``multiplicative chaos''
construction going back to Kahane's work \cite{Kah85}; also see \cite
{RV13} and references therein for recent research in that direction
which was motivated, in particular, by Quantum Gravity applications. In
two spatial dimensions, the most famous example of the random field of
that type is the two-dimensional Gaussian Free Field \cite{GFF}. A
regularized version of this field appeared in a nontrivial way in the
work of Rider and Vir\'ag \cite{VR97}, who showed that it describes
the limiting law of the log-modulus of characteristic polynomials in
the Ginibre ensemble. The Gaussian Free Field also appeared more
recently as the limiting distribution of the eigenvalue counting
function in general $\beta$-Jacobi ensembles and their principal
sub-minors \cite{BG13}. As for the one-dimensional processes with
logarithmic correlations, they are known in natural sciences under the
general name of \textit{$1/f$ noises} (see Section~2 in \cite{FlDR12}
for some general references) since, in the spectral representation, the
Fourier transform of the covariance or structure function, interpreted
as a ``power'' of the signal, is inversely proportional to the Fourier
variable (i.e., the ``frequency'' $f$). The random process $V(\theta)$
is, arguably, the simplest time-periodic stationary version of $1/f$
noise. It was found to play an important role in the construction of
conformally invariant planar random curves \cite{AJKS11} and
statistical mechanics of disordered systems \cite{FB08}. We note in
passing that from a different angle, discrete sequences with $1/f$
properties were considered heuristically in the physics literature;
see, for example, \cite{Rel04} and \cite{MLD07}.

The motivation for the work in \cite{HKOC01} came from number theory,
as for large $N$, $p_N(\theta)$ provides a good model for describing
statistics of the values of the Riemann-zeta function high up the
critical line \cite{KS00}. The established relation of $p_N(\theta)$ to
$V(\theta)$ turned out to be fruitful. It allowed one to put forward
nontrivial conjectures about statistics of extreme and high values of
characteristic polynomials of Haar unitary matrices emerging as $N\to
\infty$, and eventually for the Riemann-zeta function \cite{FHK12,FK12}.

The main goal of this paper is to investigate further the relation
between $1/f$-noises and the characteristic polynomials of random
matrices in the limit $N\to\infty$. Significantly extending the
picture found in \cite{HKOC01}, we will show that the limiting process
depends on the \textit{spectral scale} at which one allows the argument
$x$ of the characteristic polynomial $\det(\mathcal{H}-xI)$ to vary. To
this end, let us remind the reader that, as is well known in random
matrix theory (see, e.g., \cite{PS11}), there exist three natural
scales in the spectra of large random matrices. One, known as the
global, or \textit{macroscopic} scale is set for the GUE by the width of
the support of the semicircle law and, in the normalization chosen in
the present paper [see (\ref{dens})] remains of the order of unity as
$N\to\infty$. Second, known as the local, or \textit{microscopic} scale
is set by the typical separation between neighbouring eigenvalues and
is, in the chosen normalization, of order $1/N$ for large $N$. Finally,
the third scale which is called \textit{mesoscopic} can be defined as
intermediate between those two.

Deferring precise statements to the next section, now we will outline
the two instances of $1/f$ noise that emerge in the limit $N\to\infty
$ for the GUE matrices. On the macroscopic scale, by adapting the
arguments of \cite{HKOC01} to our setting, we prove that, as $N\to
\infty$, the process $\{D_N(x): x\in(-1,1)\}$ converges, after
centering, to the (aperiodic) $1/f$ noise given by the random
Chebyshev--Fourier series
%
%
\begin{equation}
\label{1fch}
F(x) = \sum_{n=1}^{\infty}
\frac{1}{\sqrt{n}} a_{n} T_{n}(x), \qquad x \in(-1,1),
\end{equation}
where $a_n$, $n=1,2 \ldots$ is a sequence of independent standard real
Gaussians. As with the Fourier series in (\ref{1f}), the convergence
in (\ref{1fch}) has to be understood in the sense of distributions in
a suitable Sobolev space. The covariance structure associated with the
generalized process (\ref{1fch}) is given by an integral operator
with kernel
$ \mathbb{E}\{F(x)F(y)\} = -\frac{1}{2}\log(2|x-y|)$.

The problem of finding a suitable model to describe the statistical
properties of the characteristic polynomials of random matrices on the
\emph{mesoscopic} rather than macroscopic scale turned out to be much
more challenging and is the main focus of the present paper. Our main
finding is the emergence of fractional Brownian motion with Hurst index
$H=0$ in this context. To describe the latter, we recall that the
conventional \textit{fractional Brownian motion} (fBm) is a zero-mean
Gaussian process $B_H(t)$, $B_H(0)=0$, with stationary increments and
the covariance structure given by
%
%
\begin{equation}
\label{intr1}
\mathbb{E} \bigl\{ \bigl[B_H(t_1)-B_H(t_2)
\bigr]^2 \bigr\}=\sigma^2 |t_1-t_2|^{2H},
\end{equation}
where $H\in(0,1)$ and $\sigma^2>0$ are two parameters. Although first
introduced by Kolmogorov in 1940, fBm became very popular after the
seminal work of Mandelbrot and van Ness \cite{ManvNess68} and proved
to be a very rich mathematical object of high utility; see, for
example, articles by M. Taqqu and by G. Molchan in the book \cite
{Taq2003} for an introduction and further references and applications.
The utility of fBm is related to its properties of being \textit{self-similar}, that is, $\{ B_H(at): t\in\mathbb{R}\}\stackrel{d}= a^H
\{ B_H(t): t\in\mathbb{R}\}$ for any $a>0$, and having \textit{stationary increments}. These two properties characterize the
corresponding Gaussian process uniquely; see, for example, \cite
{Taq2003}. In the context of self-similarity, parameter $H$ is also
known as the Hurst index $H$ or the scaling exponent.

For $H=1/2$, the fBm $B_{1/2}(t)$ is proportional to the usual Brownian
motion (Wiener process). We will denote the latter simply as $B(t)$,
with $B(dt)$ being the corresponding white noise measure, $\mathbb
{E} \{B(dt) \}=0$ and $\mathbb{E} \{B(dt)\times\break B(dt')
\}=\delta(t-t')\,dt\,dt'$, where we have chosen the normalization
corresponding to the choice of $\sigma=1$ in (\ref{intr1}).

It is apparent from (\ref{intr1}) that the naive limit $H= 0$ of
$B_H(t)$ is not well defined. To overcome this problem, the first
author proposed some time ago to regularize the fBm in the limit $H\to
0$ as follows. Consider the stochastic Fourier integral
%
%
\begin{eqnarray}
B^{(\eta)}_H(t) &=& \frac{1}{2\sqrt{2}}\int
_0^{\infty}\frac
{e^{-\eta s}}{s^{1/2+H}} \bigl[
\bigl(e^{-its}-1 \bigr)B_c(ds)+ \bigl(e^{its}-1
\bigr)\overline{B_c(ds)} \bigr],
\nonumber
\\[-8pt]
\label{intr3}
\\[-8pt]
\eqntext{\displaystyle\eta\ge0,}
\end{eqnarray}
where $B_c(t)=B_R(t)+iB_I(t)$ and $B_R(t)$ and $B_I(t)$ are two
independent copies of the Brownian motion. For $H\in(0,1)$ the
integral in (\ref{intr3}) is well defined for all $\eta\ge0$ and
represents a zero-mean Gaussian process with stationary increments and
covariance
$\mathbb{E} \{[B^{(\eta)}_H(t_1)-B^{(\eta)}_H(t_2)]^2 \}
=2\phi^{(\eta)}_H(t_1-t_2) $, where
%
%
\begin{eqnarray}
\qquad\phi^{(\eta)}_H(t)&=&\frac{1}{2}\int
_0^{\infty}\frac{e^{-2\eta
s}}{s^{1+2H}} \bigl(1-\cos{(ts)}
\bigr)\, ds
\nonumber
\\[-8pt]
\label{intr5}
\\[-8pt]
\nonumber
&=&\frac{1}{4H}\Gamma(1-2H) \biggl[ \bigl(4\eta^2+t^2
\bigr)^H\cos{ \biggl(2 H \arctan{\frac{t}{2\eta}} \biggr)}-(2
\eta)^{2H} \biggr].
\end{eqnarray}

For\vspace*{1pt} fixed $H\in(0,1)$, $\lim_{\eta\to0} \phi^{(\eta)}_H(t) =\frac
{1}{4H} \Gamma(1-2H)\cos(\pi H) t^{2H}$, where $\Gamma(z)$ is the
Euler gamma-function. Hence, $B^{(0)}_H(t)$ is fBm. This also follows
from the so-called \textit{harmonizable representation} of the fBm, which
is precisely the integral on the RHS in (\ref{intr3}) when $\eta=0$;
see Proposition~9.2 in \cite{Taq2003}, or equation (7.16) in \cite{Sam06}.
On the other hand, for any fixed $\eta>0$, the limit of $H=0$ in (\ref
{intr3}) is well defined, and
%
%
\begin{equation}
\label{covphi}
\lim_{H\downarrow0} \phi^{(\eta)}_H(t)
= \frac{1}{4}\log\biggl(\frac{t^{2}}{4\eta^{2}}+1 \biggr).
\end{equation}
We consider the resulting limiting process
%
%
\begin{equation}
\label{intr3h0}
\hspace*{6pt}B^{(\eta)}_{0}(\tau) = \frac{1}{2\sqrt{2}}\int
_{0}^{\infty}\frac
{e^{-\eta s}}{\sqrt{s}} \bigl\{
\bigl[e^{-i\tau s}-1 \bigr]B_{c}(ds)+ \bigl[e^{i\tau
s}-1
\bigr]\overline{B_{c}(ds)} \bigr\}
\end{equation}
as the most natural extension of the standard fBm to the case of zero
Hurst index $H=0$. This process can also be defined axiomatically.
\setcounter{footnote}{1}
\begin{Def*}
The regularized fBm with Hurst index $H=0$ is a
real-valued stochastic process $\{B_{0}^{(\eta)}(\tau), \tau\in
\mathbb{R} \}$ with the following properties:
\begin{longlist}[(iii)]
\item[(i)] $B_{0}^{(\eta)}(t)$ is a Gaussian process with mean 0 and
$B_{0}^{(\eta)}(0)=0$,
\item[(ii)] $\operatorname{Var}\{B_{0}^{(\eta)}(t)\}= \frac
{1}{2}\log
(\frac{t^{2}}{4\eta^{2}}+1 )$ for some $\eta>0$,
\item[(iii)] $B_{0}^{(\eta)}(t)$ has stationary increments.
\end{longlist}
\end{Def*}
The increment structure of $B_{0}^{(\eta)}(t)$ depends logarithmically
on the time separation:
%
%
\begin{equation}
\label{logcov}
\mathbb{E} \bigl\{ \bigl[B^{(\eta
)}_{0}(t_1)-B^{(\eta)}_{0}(t_2)
\bigr]^{2} \bigr\}= \frac
{1}{2}\log\biggl[\frac{(t_1-t_2)^{2}}{4\eta^{2}}+1
\biggr],
\end{equation}
and hence the regularized fBm with $H=0$ defines a bona fide
version of the $1/f$ noise with stationary increments.\footnote
{Compare (\ref{intr3h0}) with a stationary version of fBm with $H=0$
proposed in equation (16) of \cite{Schmitt2003}.} Therefore, the
stochastic process $B^{(\eta)}_{0}(\tau)$ is of interest in its own
right and deserves further study. We do not pursue this direction in
the present paper except for noting for future reference that the
regularized fBm has continuous sample paths.

\subsection*{Note}
After posting the initial version of this paper to the
arXiv, we learned of the work \cite{U09}, where a regularization of
fBm essentially equivalent to our $B_{H}^{(\eta)}(t)$ was introduced
for $H>0$. Note that neither the limit $H \to0$ nor the connection
with random matrices were identified or investigated there.

\section{Main results}
\label{semainresults}

\subsection{Macroscopic regime}\label{sec2.2}
We start with the simpler case of the macroscopic scale where we extend
the analogous construction of \cite{HKOC01} from unitary to Hermitian
matrices. The relation between characteristic polynomials of Haar
unitary matrices and the random Fourier series in (\ref{1f}) can be
understood by expanding $\log{|p_N(\theta)|}$ into the Fourier series
%
%
\begin{equation}
\label{diashah}\hspace*{6pt}
V_N(\theta)= -2\log\bigl|\det\bigl(I-Ue^{-i\theta}
\bigr)\bigr| = \sum_{n=1}^{\infty
}\frac{1}{\sqrt{n}}
\bigl(v_{n,N}e^{i n\theta}+\overline{v_{n,N}}e^{-in\theta}
\bigr),
\end{equation}
where $v_{n,N} = \frac{1}{\sqrt{n}}\operatorname{Tr} (U^{-n})$.
Now, the
coefficients $v_{n,N}$ converge in distribution as $N\to\infty$ to
independent standard complex Gaussians. This is a result due to
Diaconis and Shahshahani \cite{DS94} from which it can be inferred
\cite{HKOC01} that (\ref{1f}) represents the limit of $V_{N}(\theta)$ in a suitable functional space.

An analogue of the Diaconis--Shahshahani result for the $N \times N$
GUE matrices $\mathcal{H}$ was obtained by Johansson \cite{Joh98}. He
proved that for any fixed $m$ the vector $ (\frac{2}{\sqrt
{n}}\operatorname{Tr}
T_{n}(\mathcal{H}) )_{n=1}^m$, with $T_{n}(x) = \cos(n\arccos
(x))$ being Chebyshev polynomials, converges, after centering, to a
collection of independent standard Gaussians in the limit $N\to\infty
$. In view of the handy identity,
%
%
\begin{equation}
\label{haag1}
-\log\bigl(2|x-y|\bigr) = \sum_{n=1}^{\infty}
\frac{2}{n} T_n(x)T_n(y), \qquad x,y\in[-1,1],
x\neq y,
\end{equation}
the desired analogue of Fourier expansion is an expansion in terms of
Chebyshev polynomials,
%
%
\begin{eqnarray}
D_N(x)& =& -\log\bigl|\det(\mathcal{H} - xI)\bigr|
= \sum
_{n=1}^{\infty}\frac{a_{n,N}}{\sqrt{n}}T_{n}(x)+N
\log2+ R_{N}(x),
\nonumber
\\[-8pt]
\label{logdetH}
\\[-8pt]
\eqntext{\displaystyle a_{n,N} = \frac{2}{\sqrt{n}}\operatorname{Tr}
T_{n}(\mathcal{H}),}
\end{eqnarray}
where the error term $R_N(x)$ is due to the eigenvalues of $\mathcal H$
outside the support $[-1,1]$ of the semicircle law. Since the
probability of finding such an eigenvalue vanishes fast as $N\to\infty$,
it can be shown that the error term does not contribute in the limit
(see the proof of Proposition~\ref{propfindim} for a more precise
statement). 
One then concludes that the natural limit of $D_N(x)$, after centering,
is given by the random Chebyshev--Fourier series (\ref{1fch}).

We will make this picture mathematically rigorous by working in a
suitable functional space. First, let us assign a formal meaning to the
series in (\ref{1fch}) and the corresponding stochastic process.
Consider the space $L^2=L^2 ((-1,1), \mu(dx))$ with $ \mu(dx) =
dx/\sqrt{1-x^2}$. The Chebyshev polynomials form an orthogonal basis
in this space, with $c_n(f)$ (\ref{chebshev}) being the coefficients
of the corresponding Chebyshev--Fourier series. For $a>0$, consider the
space $V^{(a)}$ of functions $f$ in $L^2$ such that $\sum
_{n=0}^{\infty} |c_n(f)|^2 (1+n^2)^a <\infty$. This is a Hilbert
space with the inner product
\[
\langle f,g \rangle_{a}=\sum_{n=0}^{\infty}
c_n(f)c_n(g) \bigl(1+n^{2}
\bigr)^{a}.
\]
Its dual, $V^{(-a)}$, is the Hilbert space of generalised functions
$F(x)=\break \sum_{n=0}^{\infty} c_n T_n(x) $ with $\Vert F \Vert_{-a}^2=\sum
_{n=0}^{\infty} |c_n|^2(1+n^2)^{-a}<\infty$. Setting here $c_0=0$ and
$c_n=a_n/\sqrt{n}$ with $a_n$, $n\ge1$, being independent standard
Gaussians, one obtains $F(x)$ of (\ref{1fch}). In such case, $\Vert F
\Vert_{-a}^2$ is finite with probability one. This defines $F(x)$ in (\ref
{1fch}) as a generalised random function (stochastic process) which
acts on a test function $f\in V^{(a)}$ in the usual way,
\[
F[f]=\sum_{n=1}^{\infty} \frac{a_n}{\sqrt{n}}c_n(f)=\langle f,F \rangle_{0}.
\]
This process is Gaussian with zero mean. Its covariance, $\mathbb{E}\{
F[f]F[g]\}$, is given by
%
%
\begin{equation}
\label{haag}
\mathbb{E} \bigl\{F[f]F[g] \bigr\} = \sum
_{n=1}^{\infty} \frac{1}{n} \int
_{-1}^1 \!\int_{-1}^1
f(x)g(y) T_n(x)T_n(y) \mu(dx)\mu(dy).
\end{equation}
%
It can be shown (see, e.g., Lemma~3.1 in \cite{GP13}) that the order
of summation and integration in (\ref{haag}) can be interchanged, and
in view of (\ref{haag1}), one obtains the covariance operator in
closed form:
\[
\mathbb{E} \bigl\{F[f]F[g] \bigr\}= -\int_{-1}^1\!
\int_{-1}^1 \frac{1}{2}\log\bigl(2|x-y|\bigr)
f(x)g(y)  \mu(dx)\mu(dy), \qquad\hspace*{-5pt} f,g\in V^{(a)}.
\]

We are now in a position to formulate our result. Consider the centered process:
%
%
\begin{equation}
\label{logdet}\qquad
\tilde{D}_{N}(x )=-\log\bigl|\det(\mathcal{H}-xI)\bigr|
+\mathbb{E}
\bigl\{\log\bigl|\det(\mathcal{H}-xI)\bigr| \bigr\}, \qquad x\in(-1,1).
\end{equation}
Since $\log|x|$ is locally integrable, $\tilde D_{N}\in V^{(-a)}$ for
every $N$.

\smallskip

%
\begin{theorem}
\label{thglobal}
For every $a>1/2$, $\tilde{D}_N (x) \Rightarrow F(x)$ in $V^{(-a)}$ as
$N\to\infty$, where $F(x)$ given by (\ref{1fch}).
\end{theorem}

Our proof of this theorem in Section~\ref{seweakconv} involves
solving at least two technical problems that did not arise in \cite
{HKOC01}. First,\vspace*{1pt} when proving convergence of the finite-dimensional
distributions of $\tilde D_{N}(x)$, we are faced with a test function
possessing square-root singularities at the edges of the spectrum,
arising from the Chebyshev--Fourier coefficients of the logarithm
outside $[-1,1]$; see Lemma~\ref{lechebycoeffs}. Most bounds and
concentration inequalities for linear statistics rely on the test
function having at least $C^{1}(\mathbb{R})$ regularity (see, e.g., \cite{PS11,LP09,AGZ09}), while ours is only $C^{1/2}(\mathbb
{R})$ (even the recent extension \cite{SW13} of such bounds to test
functions from the $C^{1/2+\varepsilon}(\mathbb{R})$ class does not
suffice here). Making use of fine asymptotics of orthogonal polynomials
and Airy functions, we prove that this linear statistic converges to
zero, a problem that did not appear in \cite{HKOC01}.

Second, when proving tightness of $(\tilde D_{N}(x))_{N=1}^{\infty}$
we need additional control over the variance of $ \operatorname
{Tr}(T_{n}(\mathcal{H}))$ for both large $N$ \textit{and} large $n$. In \cite{HKOC01},
the analogous quantity, namely $\operatorname{Var}\{ \operatorname
{Tr}(U^{-n}) \}$, was known
explicitly due to exact results for the unitary group obtained by
Diaconis and Shashahani \cite{DS94}. In contrast, for the GUE case,
$\operatorname{Var}\{ \operatorname{Tr}(T_{n}(\mathcal{H}))\}$ and
related quantities need to be
estimated asymptotically as $N \to\infty$, \textit{uniformly} in the
degree $n$ of the Chebyshev polynomial.

\subsection{Mesoscopic regime}
Now we proceed to our next task of extending the relation between
characteristic polynomials of random matrices and $1/f$-noises to the
mesoscopic scale. In this case, instead of working directly with a
generalised stochastic process, we find it more convenient to work with
their \textit{regularized} versions.

To formulate\vspace*{1pt} our results more precisely, fix a parameter $\eta>0$ and
consider the following sequence of stochastic processes $\{W^{(\eta
)}_{N}(\tau): \tau\in\mathbb{R}\}$, $N=1,2, \ldots:$
%
%
\begin{eqnarray}
W^{(\eta)}_{N}(\tau) &=& -\log \biggl|\det\biggl[
\mathcal{H}- \biggl(x_{0}- \frac{\tau}{d_{N}} \biggr)I-
\frac
{i\eta}{d_{N}} I \biggr] \biggr|
\nonumber
\\[-8pt]
\label{wntau}
\\[-8pt]
\nonumber
&&{}+\log \biggl|\det\biggl[ \mathcal{H}-x_{0}I-
\frac{i\eta}{d_{N}}I \biggr] \biggr|.
\end{eqnarray}
Note that $ W^{(\eta)}_{N}(\tau)$ also depends implicitly on three
additional parameters: $\eta>0$, $x_{0}\in(-1,1)$ and $d_{N}>0$;
their importance is explained below, though for ease of notation we
will not emphasize the dependence on $x_0$ when referring to $ W^{(\eta
)}_{N}(\tau)$. We use the parameter $d_{N}>0$ to zoom into the
appropriate spectral scale of $\mathcal{H}$ centered around a point
$x_{0}$ inside the bulk of the limiting spectrum of the GUE matrices
${\mathcal H}$. On the macroscopic scale $d_N=1$, on the microscopic scale
$d_N=N$ whilst on the mesoscopic scale $d_N$ is in between these two
extremes, \mbox{$1 \ll d_N \ll N $}. The parameter $\eta$ is an arbitrary but
fixed positive real number, introduced to regularize the logarithmic
singularity at zero.

Our main result shows that in the \textit{mesoscopic limiting regime} where
%
%
\begin{equation}
\label{M2}
d_{N} \to\infty \quad  \mbox{and} \quad d_N=o(N/\log N) \qquad  \mbox{as }N \to\infty
\end{equation}
the stochastic process $ W^{(\eta)}_{N}(\tau)$ converges, after
centering, to $B_{0}^{(\eta)}(\tau)$; the regularized fractional
Brownian motion with Hurst index $H=0$. For finite-dimensional
distributions this is the content of the following theorem. Let
\[
\tilde{W}^{(\eta)}_{N}(\tau)= W^{(\eta)}_{N}(
\tau)-\mathbb{E} \bigl\{ W^{(\eta)}_{N}(\tau) \bigr\}.
\]

%
\begin{theorem}
\label{thmaintheorem}
Consider GUE random matrices ${\mathcal H}$ in
(\ref{dens}).
Assume that the reference point $x_0$ is in the bulk of the limiting
spectrum of ${\mathcal H}$, $x_0\in(-1,1)$, and the scaling factor $d_N$
satisfies (\ref{M2}). Then for any fixed $\eta>0$ and any finite
number of times $(\tau_{1},\ldots,\tau_{m}) \in\mathbb{R}^{m}$ we
have the convergence in distribution
%
%
\begin{equation}
\label{mesoconv}\qquad
\bigl(\tilde{W}^{(\eta)}_{N}(
\tau_{1}),\ldots,\tilde{W}^{(\eta
)}_{N}(
\tau_{m}) \bigr) \stackrel{d} {\Longrightarrow} \bigl(B^{(\eta)}_{0}(
\tau_{1}),\ldots,B^{(\eta)}_{0}(
\tau_{m}) \bigr) \qquad \mbox{as $N \to\infty$}.
\end{equation}
\end{theorem}

We prove this theorem in Section~\ref{semesoproof} by adopting
Krasovsky's derivation of identity (\ref{krasovasympt}) to the
mesoscopic scale. The characteristic function of the random vector on
the LHS in (\ref{mesoconv}) is given by a Hankel determinant whose
symbol possesses Fisher--Hartwig singularities. The Riemann--Hilbert
problem provides a powerful tool to obtain asymptotics of such Hankel
determinants \cite{De99,KV03,KMVaV04,K07}. On the mesoscopic scale the
Fisher--Hartwig singularities [these are located at points $x_0+(\tau
_k+i\eta)/d_N$] are all at distance of order $1/d_N$ from the point
$x_0 \in(-1,1)$. Because of this, the system of contours defining the
Riemann--Hilbert problem (inside of which the symbol is analytic) close
onto the real line as $N\to\infty$. In this regime, the estimates
become more delicate. In contrast, in the macroscopic regime the
Fisher--Hartwig singularities are real and spaced out and one does not
need to consider the case of shrinking contours.

Here, it is appropriate to mention that linear eigenvalue statistics on
the mesoscopic scale are more challenging to study compared to the
macroscopic scale. Known results are sparse and mostly limited to
regular test functions with compact support; see \cite
{BdMK99a,BdMK99b,S00} and also more recent works \cite
{EK13a,EK13b,DJ13,Bou14,Dui14}. One reason is that the majority of
concentration inequalities involving derivatives, such as, for example, Lipschitz norm \cite{AGZ09} or the Poincar\'e inequality
\cite{AGZ09,PS11} that proved to be so useful on the macroscopic
scale, get a factor of $d_N$ in the mesoscopic case, and hence, no
longer apply without appropriate modification. In this context, the
Riemann--Hilbert problem proves to be a powerful tool for estimating
the error terms down to very small scales (\ref{M2}).


One can extend Theorem~\ref{thmaintheorem} to an infinite-dimensional
setting with a little bit more work. Let $L^{2}[a,b]$ denote the
Hilbert space of square integrable functions on $[a,b]$ with the inner product
%
%
\begin{equation}
\langle f, g\rangle_{2} = \int_{a}^{b}f(
\tau)\overline{g(\tau)} \,d\tau.
\end{equation}
Since\vspace*{1pt} the sample paths of $\tilde{W}^{(\eta)}_{N}$ are continuous,
$\Vert \tilde{W}^{(\eta)}_{N}\Vert_{2} < \infty$. Therefore, both
$W^{(\eta)}_{N}$ and its $N \to\infty$ limit $B^{(\eta)}_{0}$ can
be viewed\vspace*{1pt} as random elements in the space $L^{2}[a,b]$. We have the following.
%

\begin{theorem}
\label{thcompactconv}
Let $-\infty< a < b < \infty$. Then on mesoscopic scales (\ref{M2}),
the process $\tilde{W}^{(\eta)}_{N}$ converges weakly (in the sense
of probability law) to $B^{(\eta)}_{0}$ in $L^{2}[a,b]$ as $N \to
\infty$. Furthermore, for every $h \in L^{2}[a,b]$, we have the
convergence in distribution
%
%
\begin{equation}
\label{hl2} \int_{a}^{b}h(\tau)
\tilde{W}^{(\eta)}_{N}(\tau)\,d\tau\stackrel{d} {\Longrightarrow}
\int_{a}^{b}h(\tau)B^{(\eta)}_{0}(
\tau)\, d\tau, \qquad N \to\infty.
\end{equation}
\end{theorem}

This result follows from Theorem $3$ in \cite{G76}, which allows one
to deduce weak convergence for general processes in $L^{2}[a,b]$ under
the hypothesis that:
\begin{longlist}[(ii)]
\item[(i)] The finite-dimensional distributions of $\tilde{W}^{(\eta
)}_{N}$ converge to those of $B^{(\eta)}_{0}$ as $N \to\infty$.
\item[(ii)] For some $C>0$, the bound $\mathbb{E}\{|\tilde{W}^{(\eta
)}_{N}(\tau)|^{2}\} \leq C$ holds for all $N$ and $\tau\in[a,b]$ and
%
%
\begin{equation}
\lim_{N \to\infty}\mathbb{E} \bigl\{\bigl|\tilde{W}^{(\eta)}_{N}(
\tau)\bigr|^{2} \bigr\} = \mathbb{E} \bigl\{\bigl|B^{(\eta)}_{0}(
\tau)\bigr|^{2} \bigr\}.
\end{equation}
\end{longlist}
Note that item (i) is a restatement of Theorem~\ref{thmaintheorem}, while item (ii) will be shown to follow
from our proof of Theorem~\ref{thmaintheorem}.

Having established the relation between characteristic polynomials of
GUE matrices and $1/f$ noise on the mesoscopic scale, let us revisit
the series expansions of the macroscopic scale\vspace*{1pt} discussed at length in
Section~\ref{sec2.2}. Instead of expanding the process $W^{(\eta
)}_{N}(\tau)$ in a Chebyshev--Fourier series and applying the
Diaconis--Shahshahani result, in the mesoscopic regime it comes in
handy to expand $W^{(\eta)}_{N}(\tau)$ as a Fourier \textit
{integral}. 

To this end, we now provide a suitable Fourier-integral representation
for $W^{(\eta)}_{N}(\tau)$. Such a representation can be derived by
making use of the identity (see, e.g., equation (7.89) in \cite{De99})
%
%
\begin{equation}
\label{deiftint}
\frac{1}{2}\log\biggl(\frac{t^{2}}{\varepsilon
^{2}}+1 \biggr) = \int
_{0}^{\infty}\frac{e^{-\varepsilon s}}{s} \bigl[1-\cos(ts)
\bigr] \,ds, \qquad \varepsilon>0.
\end{equation}
It follows from (\ref{deiftint}) that
%
%
\begin{equation}
\label{detforident} W^{(\eta)}_{N}(\tau) = \frac{1}{2}\int
_{0}^{\infty}\frac{e^{-\eta
s}}{\sqrt{s}} \bigl\{
\bigl[e^{-i\tau s}-1 \bigr]b_{N}(s)+ \bigl[e^{i\tau s}-1
\bigr]  \overline{b_{N}(s)} \bigr\}\,ds,
\end{equation}
where
%
%
\begin{equation}
\label{bee} b_{N}(s) = \frac{1}{\sqrt{s}} \operatorname
{Tr}e^{-isd_{N}(\mathcal
{H}-x_{0}I)}.
\end{equation}
The identity (\ref{detforident}) can be thought of as the Fourier
integral version of the Fourier series (\ref{diashah}). Furthermore,
comparison of the harmonizable representation (\ref{intr3h0}) for
$B_0^{(\eta)}(t)$ [which can be thought as a natural integral analogue
of the series expansions in (\ref{1f}) and (\ref{detforident})],
suggests that the Fourier coefficients $b_{N}(s)$ converge in the
mesoscopic regime to Gaussian white noise. Such a statement may be
interpreted as a continuous analogue of the Diaconis--Shahshahani
result \cite{DS94} and is the content of our next theorem. 

Let $C^{\infty}_{0}(\mathbb{R}_{+})$ be the space of infinitely many
times differentiable functions with compact support on $\mathbb
{R}_{+}=\{x\in\mathbb{R}: x>0 \}$. Denote
%
%
\begin{equation}
\label{cn} c_{N}(\xi) = \int_{0}^{\infty}
\xi(s)b_{N}(s)\, ds.
\end{equation}

%
\begin{theorem}
\label{thfourierconv}
Consider the mesoscopic regime where $d_N=N^{\alpha}$ with any $\alpha
\in(0,1)$. Then for every $\xi\in C^{\infty}_0 (\mathbb{R}_{+})$
%
%
\begin{equation}
\label{whitenoiselimit}
\lim_{N \to\infty}\mathbb{E} \bigl\{
e^{-i\operatorname{Re}c_{N}(\xi)}
\bigr\} = \exp\biggl(-\frac{1}{4}\int_{0}^{\infty}\bigl|
\xi(s)\bigr|^{2}\, ds \biggr).
\end{equation}
Furthermore, for any finite number of $\xi_j \in C^{\infty}_0
(\mathbb{R}_{+})$, the vector $(c_{N}(\xi_{1}),\ldots,\break  c_{N}(\xi
_{m}))$ converges in distribution, as $N \to\infty$, to the centered
complex Gaussian vector $Z \in\mathbb{R}^{m}$ having relation matrix
$\mathbb{E}(ZZ^{\mathrm{T}})=0$ and covariance matrix $\Gamma=
\mathbb{E}(ZZ^{\dagger})$ given by
%
%
\begin{equation}
\Gamma_{j,k} = \int_{0}^{\infty}
\xi_{j}(s)\overline{\xi_{k}(s)}\, ds, \qquad j,k=1,
\ldots,m.
\end{equation}
\end{theorem}

\begin{pf}
See Section~\ref{sefouriercoeff}.\noqed
\end{pf}

%
\begin{remark}
As is often the case in random matrix theory, linear eigenvalue
statistics such as (\ref{cn}) have variance of the order of unity due
to strong correlations between the eigenvalues and converge to a
Gaussian random variable after centering. One would typically expect
that $\mathbb{E}\{c_{N}(\xi)\} = O(N/d_{N})$ as $N \to\infty$.
Instead, we find (see Section~\ref{sefouriercoeff}) that the
smoothness of $\xi$ and the rapid oscillations in (\ref{bee}) imply
$\mathbb{E}\{c_{N}(\xi)\} = O(d_{N}^{-1})$ as $N \to\infty$, and
thus, centering is not really needed.
\end{remark}

The rest of the paper is organized as follows.
Section~\ref{semesoproof} is devoted to the proof of Theorem~\ref{thmaintheorem}.
To do this, we begin by adapting the differential
identity used in \cite{K07} and then outline the relevant asymptotic
analysis of the Riemann--Hilbert problem, leaving estimation of all
error terms to Appendix~\ref{apriemann}. 
Section~\ref{sefouriercoeff} is devoted to proving the convergence of
the Fourier coefficients $b_{N}(s)$ to the white noise. In the final
section, we focus on the macroscopic scale and prove Theorem~\ref{thglobal}.

\section{Mesoscopic regime}
\label{semesoproof}
In this section, we prove Theorem~\ref{thmaintheorem}. Let us fix
\mbox{$m-1$} distinct times $\tau_1, \ldots, \tau_{m-1}$, $m\ge2$, and
consider the characteristic function
\[
\varphi_N ( \alpha_1, \ldots, \alpha_{m-1})=
\mathbb{E} \Biggl\{ \exp\Biggl(\sum\limits
_{k=1}^{m-1}
\alpha_{k} W^{(\eta)}_{N}(\tau_{k})
\Biggr) \Biggr\}
\]
of the random vector $(W^{(\eta)}_{N}(\tau_1), \ldots, W^{(\eta
)}_{N}(\tau_{m-1}))$. Our strategy will be to prove that $\varphi_N$
converges to the characteristic function of the multivariate Gaussian
distribution in the limit $N\to\infty$. Theorem~\ref{thmaintheorem}
will then follow by inspection of the quadratic form in the exponential.

To begin with, we will write the characteristic function $\varphi_N$
as the partition function of a matrix model with Gaussian weight,
modified by the singularities
%
%
\begin{equation}
\mu_{k}= \sqrt{2N} \biggl(x_{0}+\frac{\tau_{k}+i\eta}{d_{N}}
\biggr), \qquad\eta>0, 
\end{equation}
where $k=1,\ldots,m$ and $\tau_{m}\equiv0$. A standard calculation
(changing variables of integration from ${\mathcal H}$ to the eigenvalues
and eigenvectors of ${\mathcal H}$ and integrating out the eigenvectors;
see, e.g., \cite{PS11}) yields
%
%
\begin{equation}\label{multint}\qquad
\varphi_N ( \alpha_1, \ldots, \alpha_{m-1}) =
\frac{1}{C}\int_{\mathbb{R}^{N}}\prod
_{j=1}^{N}w(x_{j})\prod
_{1 \leq i < j \leq
N}(x_{i}-x_{j})^{2}
\,dx_{1}\cdots dx_{N},
\end{equation}
where the weight function is given by
%
%
\begin{equation}
\label{weight} w(x) = e^{-x^{2}}\prod_{k=1}^{m}|x-
\mu_{k}|^{\alpha_{k}}, \qquad\operatorname{Im}(\mu_{k})
\neq0, \qquad k=1,\ldots,m
\end{equation}
and $\alpha_m=-\alpha_{1} - \cdots- \alpha_{m-1}$. Note the
discrepancy with the measure (\ref{dens}); for convenience we have
changed variables $x_{j} \to x_{j}/\sqrt{2N}$, the resulting
multiplicative constants cancelling each other out.

Our calculation will be guided by that of Krasovsky \cite{K07} who
treated a similar partition function, but only for the macroscopic
regime $d_{N}=1$ and $\eta=0$. In that case, the weight function
acquires Fisher--Hartwig singularities inside the spectral interval
$(-1,1)$. In contrast, our weight (\ref{weight}) possesses
singularities in the complex plane that merge toward the point $x_{0}$
on the spectral axis at rate $d_{N}$ as $N\to\infty$. Since this
merging process occurs sufficiently slowly [i.e., $d_{N} = o(N)$],
these singularities will not play a crucial role in the calculation.

A special feature of the weight function (\ref{weight}) is the \textit
{cyclic condition}
%
%
\begin{equation}
\label{cyclic} \sum_{k=1}^{m}
\alpha_{k}=0.
\end{equation}
This holds because the second term in (\ref{wntau}) is independent of
$\tau$. Our first step is to express the partition function (\ref
{multint}) in a form suitable for the computation of asymptotics.

\subsection{Orthogonal polynomials and differential identity}
The multiple integral in (\ref{multint}) is intimately connected to
the theory of orthogonal polynomials. Let
\[
p_{n}(x) = \chi_{n} \bigl(x^{n}+
\beta_{n}x^{n-1}+\gamma_{n}x^{n-2}+\cdots
\bigr), \qquad  n=0,1,2, \ldots,
\]
be orthogonal polynomials with respect to weight function\vspace*{1pt} 
$w(x)$:\break 
$\int_{-\infty}^{\infty}p_{m}(x) p_{n}(x)w(x)\,dx=\delta_{m,n}$.
When the $\alpha_{j}$'s are real and each $\alpha_{j}>-1/2$ we have
$w(x)\geq0$ and the existence of the polynomials $p_{n}(x)$ is well
known \cite{De99}. Then, as in \cite{K07}, the coefficients $\chi
_{n},\beta_{n}$ and $\gamma_{n}$ and the polynomials $p_{n}(x)$ are
defined for any $\{\alpha_{j}\}_{j=1}^{m} \in\mathbb{C}^{m}$ via
analytic\vspace*{1pt} continuation, provided each $\operatorname{Re}(\alpha_{j})>-1/2$.

Now, the partition function (\ref{multint}) can be written in terms of
the coefficients $\{\chi_{j}\}_{j=1}^{N}$ (see, e.g., \cite{Meh04})
%
%
\begin{equation}
\label{prodchi}
\varphi_N (\alpha_{1},\ldots,
\alpha_{m-1}) = \frac{N!}{C}\prod_{j=0}^{N-1}
\chi_{j}^{-2}.
\end{equation}
Thus, in principle, our problem is reduced to computing the asymptotics
of the orthogonal polynomials and related quantities with respect to
the weight $w(x)$.
The crucial point observed in \cite{K07} is that by taking the
logarithmic derivative on both sides of (\ref{prodchi}) with respect
to any of the $\alpha_{j}$'s, the RHS can be written as a sum
involving only $O(m)$ terms, rather than $N$. To state the resulting
\textit{differential identity} we also need the following $2 \times2$
matrix involving the orthogonal polynomials and their Cauchy transforms:
%
%
\begin{equation}
\label{ymatrix}
Y(z) =
\pmatrix{ \displaystyle\chi_{N}^{-1}p_{N}(z)
& \displaystyle{\chi_{N}^{-1}\int_{-\infty}^{\infty}
\frac{p_{N}(x)}{x-z}\frac
{w(x)\,dx}{2\pi i}}\vspace*{6pt}
\cr
\displaystyle{-2\pi i
\chi_{N-1}p_{N-1}(z)} &\displaystyle{ -\chi_{N-1}\int
_{-\infty}^{\infty}\frac{p_{N-1}(x)}{x-z}w(x)\,dx} }.
\end{equation}

%
\begin{lemma}
For each $k=1,\ldots,m$, let $\mu_{k}$ in (\ref{weight}) be any
complex parameters satisfying $\operatorname{Im}(\mu_{k}) \neq0$ and
define $\alpha_{m+k}=\alpha_{k}$, $\mu_{m+k} = \overline{\mu
_{k}}$. Denoting by $'$ differentiation with respect to $\alpha_{j}$,
the following formula holds for any $j=1,\ldots,m$:
%
%
\begin{eqnarray}
(\log\varphi_N)'  &=&  -N(\log
\chi_{N}\chi_{N-1})'-2 \biggl(\frac
{\chi_{N-1}}{\chi_{N}}
\biggr)^{2} \biggl(\log\frac{\chi
_{N-1}}{\chi_{N}} \biggr)^{\prime}+2
\bigl( \gamma_{N}'-\beta_{N}\beta_{N}'
\bigr)
\nonumber
\\
\label{diffid}
&&{}+\frac{1}{2}\sum_{k=1}^{2m}
\alpha_{k} \bigl(Y_{11}(\mu_{k})'Y_{22}(
\mu_{k})-Y_{21}(\mu_{k})'Y_{12}(
\mu_{k})\\
&&{}+(\log\chi_{N}\chi_{N-1})'Y_{11}(
\mu_{k})Y_{22}(\mu_{k}) \bigr).\nonumber
\end{eqnarray}
\end{lemma}

\begin{pf}
The proof follows from simple modifications of the arguments given in
Section~3 of \cite{K07}. In fact, further simplifications occur due to
the cyclic condition $\sum_{k=1}^{m}\alpha_{k}=0$ and the fact that
the singularities $\mu_{k}$ have nonzero imaginary part ($k=1,\ldots,m$).
\end{pf}

Note that $\chi_{N}$ and the coefficients $\beta_{N}$ and $\gamma
_{N}$ can be computed from the relations:
%
%
\begin{eqnarray}
Y_{11}(z) &=& z^{N}+
\beta_{N}z^{N-1}+\gamma_{N}z^{N-2}+\cdots,
\nonumber
\\[-8pt]
\label{coeffid}
\\[-8pt]
\nonumber
\chi_{N-1}^{2} &= &\lim_{z \to\infty}
\frac{iY_{21}(z)}{2\pi z^{N-1}}.
\end{eqnarray}
Therefore, our plan will be to compute the asymptotics of $Y(z)$ and
then, by making use of identities (\ref{coeffid}), evaluate the RHS of
(\ref{diffid}) to the desired accuracy in the limit as $N \to\infty
$. We will find that the error terms in the asymptotics are uniform in
the variables $\{\alpha_{k}\}_{k=1}^{m-1}$ belonging to a compact
subset of
%
%
\begin{equation}\label{omset}
\Omega= \bigl\{(\alpha_{1},\ldots,\alpha_{m-1}) \vert\operatorname{Re}(
\alpha_{k}) > -1/2,  k=1,\ldots,m-1 \bigr\}.
\end{equation}
This uniformity property then allows us to integrate the identity
(\ref{diffid}) recursively with respect to $\{\alpha_{k}\}
_{k=1}^{m-1}$ and obtain asymptotics for the characteristic function~(\ref{multint}). The asymptotics of $Y(z)$ in the limit $N\to\infty$
can be obtained by using an appropriate Riemann--Hilbert problem.
Although this technique is nowadays standard, for the reader's
convenience we will briefly summarise the necessary ingredients of the
corresponding calculation.

\subsection{The Riemann--Hilbert problem for $Y(z)$}

The relationship between orthogonal polynomials and Riemann--Hilbert
problems was established for general weights in \cite{FIK92} where it
was shown that $Y(z)$ solves the following\vadjust{\goodbreak} problem:
\begin{longlist}[3.]
\item[1.] $Y(z)$ is analytic in $\mathbb{C} \setminus\mathbb{R}$.
\item[2.] On the real line there is a jump discontinuity
%
%
\begin{equation}
Y_{+}(x) = Y_{-}(x) %
\pmatrix{ 1 & w(x)
\vspace*{2pt}
\cr
0 & 1 }, \qquad x \in\mathbb{R},
\end{equation}
where $Y_{+}(x)$ and $Y_{-}(x)$ denote the limiting values of $Y(z)$ as
$z$ approaches the point $x \in\mathbb{R}$ from above $(+)$ or below
$(-)$.
\item[3.] Near $z=\infty$, we have the following asymptotic behaviour:
%
%
\begin{equation}\label{Yasympt}
Y(z) = \biggl(I+O \biggl(\frac{1}{z} \biggr) \biggr)z^{N\sigma_{3}}.
\end{equation}
\end{longlist}
Here, $\sigma_{3}$ is the third Pauli matrix and serves as a
convenient notational tool. By definition of the matrix exponential,
the notation in (\ref{Yasympt}) has the meaning
%
%
\begin{equation}
z^{N\sigma_{3}} = %
\pmatrix{ z^{N} & 0\vspace*{2pt}
\cr
0 &
z^{-N} }.
\end{equation}

One can verify directly that $Y(z)$ of (\ref{ymatrix}) does indeed
solve this Riemann--Hilbert problem, while the uniqueness of this
solution can be deduced from the observation that $\det Y(z) \equiv1$,
in conjunction with the Liouville theorem. Further details regarding
existence and uniqueness of the problem can be found in \cite{De99}.

In order to obtain asymptotics as $N \to\infty$, we will perform a
sequence of transformations to our initial Riemann--Hilbert problem
known as the \textit{Deift--Zhou steepest descent} (see, e.g., \cite
{De99} and \cite{DKMVZ99}). The purpose of these transformations is to
identify a ``limiting'' problem that can be solved with elementary
functions, giving the leading order asymptotics to $Y(z)$. For the
reader's convenience, we briefly describe the key points underlying
these transformations:
\begin{longlist}[3.]
\item[1.] The first transformation $Y \to T$ normalizes the unsatisfactory
asymptotic behaviour in the third condition, equation (\ref{Yasympt}).
This comes with the cost that the entries of the jump matrix for $T(z)$
on the interval $(-1,1)$ are now oscillating in $N$ and do not have a
limit as $N \to\infty$.
\item[2.] The second transformation $T \to S$ aims to remove these
oscillations by splitting the contour $(-1,1)$ into lens shaped
contours where now the jump matrices are exponentially close to the
identity. For our particular \textit{mesoscopic} problem, we need the
lenses to pass below the singularities for each $k=1,\ldots,m$, so
that their distance from $(-1,1)$ is of order $O(d_{N}^{-1})$ (see
Figure~\ref{figcontour}).
\item[3.] Now it turns out that the jump matrices for $S$ tend to the
identity as $N \to\infty$, except on the contour $(-1,1)$. But the
jump across $(-1,1)$ is of a special form that can be solved exactly in
terms of elementary functions. This solution, denoted $P_{\infty}(z)$,
gives the leading order contribution to the asymptotics in the required
regions of the complex plane.
\end{longlist}

%
\begin{figure}

\includegraphics{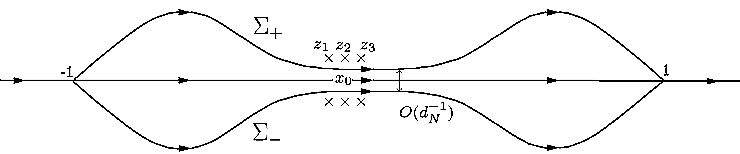}

\caption{The contour $\Sigma$ for the $S$ Riemann--Hilbert problem
with $m=3$. The crosses depict the $3$ singularities and their complex
conjugates, of distance $O(d_{N}^{-1})$ from the point $x_{0} \in
(-1,1)$. The lenses $\Sigma_{\pm}$ pass between the real line and the
singularities into the points $\pm1$.}
\label{figcontour}
\end{figure}

In Section~\ref{seasymptotics}, we will show that the asymptotics
obtained in this way lead directly to Theorem~\ref{thmaintheorem}.
However, to complete the proof, one has to show that the conclusion of
$(3)$, namely that $S(z) \sim P_{\infty}(z)$ as $N \to\infty$, is
really correct. This may be regarded as the most technical part of the
Deift--Zhou method. The main problem is that although the jump matrix
for $S(z)$ converges to that of $P_{\infty}(z)$, this convergence is
not uniform near the edges $z=\pm1$. To remedy this, local solutions
known as \textit{parametrices} have to be constructed near these
points, and then matched to leading order with the so-called outer
parametrix $P_{\infty}(z)$. These final technical issues will be
addressed in Appendix~\ref{apriemann}.

\subsection{$T$ and $S$ transformations of the Riemann--Hilbert problem}
\label{sesmatrix}
The $T$ transformation is performed in the usual way. First, we define
the $g$-function:
%
%
\begin{equation}
\label{gfun}
g(z) = \int_{-1}^{1}\log(z-s)
\rho(s)\,ds, \qquad z \in\mathbb{C}\setminus(-\infty,1],
\end{equation}
where throughout we take the principal branch of the logarithm. Here
and below, $\rho(s) = (2/\pi)\sqrt{1-s^{2}}$ denotes the limiting
density of eigenvalues. The $Y \to T$ transformation is then given by
the formula
%
%
\begin{equation}
Y(z\sqrt{2N}) = (2N)^{N\sigma_{3}}e^{Nl\sigma
_{3}/2}T(z)e^{N(g(z)-l/2)\sigma_{3}},
\end{equation}
where $l=-1-2\log(2)$. Notice that we have rescaled the
Riemann--Hilbert problem so that the singularities of the corresponding
weight function are of order $O(1)$ as $N \to\infty$, so that from
now on we deal with singularities of the form
%
%
\begin{equation}\label{sings}
z_{k} = \frac{\mu_{k}}{\sqrt{2N}} = x_{0}+\frac{\tau_{k}+i\eta
}{d_{N}}.
\end{equation}

The resulting jump matrix for $T(z)$ can now be computed from the
standard properties of the $g$-function:
%
%
\begin{eqnarray}
 g_{+}(x) + g_{-}(x)-2x^{2}-l & = &0,
\qquad x \in(-1,1),
\nonumber\\
 \label{gprops}
g_{+}(x) + g_{-}(x)-2x^{2}-l &<& 0, \qquad x
\in\mathbb{R}\setminus[-1,1],
\\
g_{+}(x)-g_{-}(x) &=&
\cases{ 2\pi i, & \quad
$x\leq-1$,\vspace*{3pt}
\cr
\displaystyle 2\pi i \int_{x}^{1}
\rho(s)\,ds, & \quad$x \in[-1,1]$,\vspace*{3pt}
\cr
0, & $\quad x \geq1$.}\nonumber
\end{eqnarray}
In addition, since $g(z) \sim\log(z)$ as $z \to\infty$, we have
$e^{Ng(z)\sigma_{3}} \sim z^{N\sigma_{3}}$. Thus, one easily verifies
that $T(z)$ is normalized at $z=\infty$. We now have the following
Riemann--Hilbert problem for $T(z)$:
\begin{longlist}[3.]
\item[1.] $T(z)$ is analytic in $\mathbb{C}\setminus\mathbb{R}$.
\item[2.] We have the jump condition
%
%
\begin{eqnarray}
T_{+}(x) &=& T_{-}(x) %
\pmatrix{ e^{-N(g_{+}(x)-g_{-}(x))} & \displaystyle{\prod
_{k=1}^{m}|x-z_{k}|^{\alpha_{k}}}
\vspace*{3pt}
\cr
0 & e^{N(g_{+}(x)-g_{-}(x))} },
\nonumber
\\[-8pt]
\label{tjump1}
\\[-8pt]
\eqntext{\displaystyle x \in(-1,1),}
\\
 T_{+}(x) &=& T_{-}(x) %
\pmatrix{ 1 & \displaystyle\prod_{k=1}^{m}|x-z_{k}|^{\alpha
_{k}}e^{N(g_{+}(x)+g_{-}(x)-2x^{2}-l)}
\vspace*{3pt}
\cr
0 & 1},
\nonumber
\\[-8pt]
\label{tjump2}
\\[-8pt]
\eqntext{ \displaystyle x \in\mathbb{R}\setminus[-1,1].}
\end{eqnarray}
\item[3.] $T(z) = I+O(z^{-1})$ as $z \to\infty$.
\end{longlist}
%
We see that although the problem for $T(z)$ is normalized at $\infty$;
the jump matrix~(\ref{tjump1}) on $(-1,1)$ has oscillatory diagonal
entries that not have a limit as $N \to\infty$. The Deift--Zhou
steepest descent procedure remedies this situation by splitting the
contour $(-1,1)$ into ``lenses'' in the complex plane (see Figure~\ref
{figcontour}), transforming the unwanted oscillations into
exponentially decaying matrix elements.

This procedure is facilitated by the factorization of the jump matrix
on $(-1,1)$:
\begin{eqnarray*}
&& \pmatrix{ e^{-Nh(x)} & \omega(x)\vspace*{2pt}
\cr
0 &
e^{Nh(x)}} \\
&&\qquad= \pmatrix{ 1 & 0\vspace*{2pt}
\cr
\omega(x)^{-1}e^{Nh(x)}
& 1} \pmatrix{ 0 & \omega(x)^{-1}\vspace*{2pt}
\cr
-
\omega(x)^{-1} & 0 } \pmatrix{ 1 & 0\vspace*{2pt}
\cr
\omega(x)^{-1}e^{-Nh(x)} & 1 },
\end{eqnarray*}
where
%
\begin{eqnarray}
\label{omdef}
\omega(x) &=& \prod_{k=1}^{m}|x-z_{k}|^{\alpha_{k}},
\\
\label{hdef} h(x) &=& g_{+}(x)-g_{-}(x) = -2\pi i \int
_{1}^{x}\rho(y)\,dy.
\end{eqnarray}

The latter objects (\ref{omdef}) and (\ref{hdef}) possess analytic
continuations into the lens shaped regions depicted in Figure~\ref{figcontour}. For the weight $\omega(x)$, we have
%
%
\begin{equation}
\label{omegaweight}
\omega(z) = 
\prod
_{k=1}^{m-1} \biggl[\frac{(z-x_{0}-\tau_{k}/d_{N})^{2}+(\eta
/d_{N})^{2}}{(z-x_{0})^{2}+(\eta/d_{N})^{2}}
\biggr]^{\alpha_{k}/2},
\end{equation}
where throughout we take the principal branch of the roots. This
function is analytic for all $z$ such that the inequality
%
%
\begin{equation}
\label{ineqanalytic}
\bigl(\operatorname{Re}(z)-\operatorname
{Re}(z_{k}) \bigr)^{2} > \bigl(
\operatorname{Im}(z_{k}) \bigr)^{2}- \bigl(\operatorname{Im}(z)
\bigr)^{2}
\end{equation}
is satisfied for every $k=1,\ldots,m$. One easily verifies that for
$x_{0} \in(-1+\delta,1-\delta)$, the inequality (\ref
{ineqanalytic}) holds for any $z$ chosen from the interior region
bounded by the lips $\Sigma_{\pm1}$ and the discs $z \in\partial
B_{\pm1}(\delta)$ of sufficiently small radius (see Figure~\ref
{figcontour}). Finally, let $h(z)$ denote the analytic continuation of
(\ref{hdef}) to $\mathbb{C}\setminus((-\infty,-1]\cup[1,\infty
))$. We are now ready to define the $T \to S$ transformation. Let
%
%
\begin{equation}
S(z) = %
\cases{ T(z), \qquad \mbox{for $z$ outside the lenses},
\vspace*{3pt}
\cr
T(z)
\pmatrix{ 1 & 0\vspace*{3pt}
\cr
-
\omega(z)^{-1}e^{-Nh(z)} & 1}, \vspace*{3pt}\cr
\qquad\qquad\hspace*{5pt}\mbox{for $z$ in the upper
part of the lenses},\vspace*{2pt}
\cr
T(z) %
\pmatrix{ 1 & 0
\vspace*{3pt}
\cr
\omega(z)^{-1}e^{Nh(z)} & 1}, \vspace*{3pt}\cr
\qquad\qquad\hspace*{5pt}\mbox{for
$z$ in the lower part of the lenses.}}
\end{equation}
Now we get the following Riemann--Hilbert problem for $S(z)$:
\begin{longlist}[3.]
\item[1.] $S(z)$ is analytic in $\mathbb{C}\setminus\Sigma$ where
$\Sigma= \Sigma_{+}\cup\mathbb{R}\cup\Sigma_{-}$.
\item[2.] $S(z)$ has the following jumps on $\Sigma$:
\begin{eqnarray*}
S_{+}(x) &=& S_{-}(x) %
\pmatrix{ 1 & 0
\vspace*{2pt}
\cr
\omega(x)^{-1}e^{\mp Nh(x)} & 1}, \qquad x \in
\Sigma_{\pm},
\\
S_{+}(x) &=& S_{-}(x) %
\pmatrix{ 0 & \omega(x)
\vspace*{2pt}
\cr
-\omega(x)^{-1} & 0 }, \qquad x \in(-1,1),
\\
S_{+}(x) &=& S_{-}(x) %
\pmatrix{ 1 &
\omega(x)e^{N(g_{+}(x)+g_{-}(x)-2x^{2}-l)}\vspace*{2pt}
\cr
0 & 1 }, \qquad x \in\mathbb{R}
\setminus[-1,1].
\end{eqnarray*}
\item[3.] $S(z) = I+O(z^{-1})$ as $z \to\infty$.
\end{longlist}

At this point in the asymptotic analysis, it becomes clear that the
mesoscopic regime under consideration becomes important. In order to
obtain asymptotics, it is essential that the jump matrix for $S(z)$
approaches the identity as $N \to\infty$ for $z\in\Sigma_{\pm}$.
In the\vspace*{1pt} \hyperref[app]{Appendix} (see Proposition~\ref{propdeltasigr}), we will see
that $|e^{\mp Nh(z)}| = O(e^{-c_{1}({N}/{d_{N}})})$ as $N \to\infty
$ uniformly on $\Sigma_{\pm}\setminus(B_{1}(\delta)\cup
B_{-1}(\delta))$. Notice that such a bound fails when one approaches
the critical situation $d_{N}=N$ corresponding to the \textit{local}
or \textit{microscopic} regime. It is precisely at this scale that one
would \textit{not} expect the appearance of a Gaussian process in the
limit $N \to\infty$.

Therefore, in the \textit{mesoscopic} regime it is reasonable to
expect that in the limit $N \to\infty$ we may neglect the jumps on
$\Sigma_{\pm}\cup(\mathbb{R}\setminus[-1,1])$ and approximate
$S(z)$ by a Riemann--Hilbert problem with jumps only on the interval
$(-1,1)$. This approximation will be valid only in the region
$U_{\infty}=\mathbb{C}\setminus(B_{1}(\delta)\cup B_{-1}(\delta))$
and will give rise to an error that is quantified in Appendix~\ref{apriemann}.

\subsection{Limiting Riemann--Hilbert problem: Parametrix in
\texorpdfstring{$U_{\infty}$}{Uinfty}}
\label{seouter}
Before we perform the final transformation $S \to R$ of the
Riemann--Hilbert problem, we must construct parametrices in the
appropriate regions of the complex plane. We saw in the last section
how the jump matrices for $S(z)$ converge to the identity as $N \to
\infty$, except on $[-1,1]$. Therefore, outside the lenses and the
discs, we expect the solution to the following problem to give a good
approximation to $S(z)$ for large~$N$:
\begin{longlist}[3.]
\item[1.] $P_{\infty}(z)$ is analytic in $\mathbb{C}\setminus[-1,1]$.
\item[2.] We have the jump condition
%
%
\begin{equation}
P_{\infty,+}(x)=P_{\infty,-}(x) %
\pmatrix{ 0 & \omega(x)
\cr
-\omega(x)^{-1} & 0}, \qquad x \in(-1,1).
\end{equation}
\item[3.] $P_{\infty}(z) = I+O(z^{-1})$ as $z \to\infty$.
\end{longlist}

This problem has the advantage that it has a completely explicit
solution. The solution, as obtained in \cite{KMVaV04}, is given by
%
%
\begin{eqnarray}
 P_{\infty}(z) &=& \frac{1}{2}(\mathcal{D}_{\infty
})^{\sigma_{3}}
\pmatrix{ a+a^{-1} & -i \bigl(a-a^{-1} \bigr)
\vspace*{2pt}
\cr
i \bigl(a-a^{-1} \bigr) & a+a^{-1}}
\mathcal{D}(z)^{-\sigma_{3}},
\nonumber
\\[-8pt]
\label{pinf}
\\[-8pt]
\nonumber
 a(z) &=& \frac{(z-1)^{1/4}}{(z+1)^{1/4}},
\end{eqnarray}
where $\mathcal{D}(z)$ is the Szeg\"o function
%
%
\begin{equation}
\label{sz}
\mathcal{D}(z) = \exp\biggl(\frac{\sqrt{z+1}\sqrt
{z-1}}{2\pi}\int
_{-1}^{1}\frac{\log\omega(x)}{\sqrt
{1-x^{2}}}\frac{dx}{z-x}
\biggr)
\end{equation}
and
%
%
\begin{equation}
\label{dinf} \mathcal{D}_{\infty} = \lim_{z \to\infty}
\mathcal{D}(z) = \exp\biggl(\frac{1}{2\pi}\int_{-1}^{1}
\frac{\log\omega(x)}{\sqrt
{1-x^{2}}}\,dx \biggr).
\end{equation}
Recalling the definition of the weight $\omega(x)$ in (\ref{omdef}),
the integrals in (\ref{sz}) can be calculated explicitly by extending
the procedure outlined in \cite{K07} to the case of complex singularities.

As we shall see in the next subsection, the Szeg\"o function $\mathcal
{D}(z)$ will turn out to be the key ingredient in deriving the
logarithmic covariance structure in (\ref{logcov}).

\subsection{Asymptotics of the polynomials and proof of
Theorem~\texorpdfstring{\protect\ref{thmaintheorem}}{2.2}}
\label{seasymptotics}
We are now ready to present the leading order asymptotics $N \to\infty
$ of the $Y$-matrix in (\ref{ymatrix}), leaving the technical matters
of estimation of errors and the final transformation of the
Riemann--Hilbert problem to Appendix~\ref{apriemann}. Our aim in this
subsection is to prove Theorem~\ref{thmaintheorem} using these asymptotics.

Tracing back the transformations $S \to T \to Y$, we find that
%
%
\begin{equation}
\label{yas} Y(z\sqrt{2N}) = (2N)^{N\sigma_{3}/2}e^{Nl\sigma
_{3}/2}S(z)e^{N(g(z)-l/2)\sigma_{3}}.
\end{equation}
%
According to (\ref{diffid}), we need the asymptotics for $Y(z)$ in two
different regions of the complex plane, near $z=\infty$ in the first
line of (\ref{diffid}) and at $z=z_{k}$ in the second line. In the
following proposition, let $\mathcal{A}$ denote the bounded subset of
$\mathbb{C}$ enclosed by the lenses $\Sigma_{\pm}$ and the discs
$\partial B_{\pm1}(\delta)$.
%

\begin{proposition}
\label{11111}
Consider the Riemann--Hilbert problems $S(z)$ and $P_{\infty}(z)$ from
Sections~\ref{sesmatrix} and \ref{seouter}, respectively. Then the
following asymptotics hold as $N \to\infty$:
%
%
\begin{equation}\qquad
S(z) = \biggl(I+\frac{\tilde{R}_{1}(z)}{N}+O \biggl(\frac
{1}{Nd_{N}} \biggr)+O \bigl(
\log(d_{N})e^{-c_{1}({N}/{d_{N}})} \bigr) \biggr)P_{\infty}(z),
\label{zkas}
\end{equation}
uniformly for all $z \in\mathbb{C}\setminus\mathcal{A}$. The
function $\tilde{R}_{1}(z)$ has an asymptotic expansion of the form
$\tilde{R}_{1}(z) = (A/z+B/z^{2}+O(z^{-3}))$ as $z \to\infty$ where
$c_{1}$ is a positive constant depending only on $\delta$ and $\eta$ and
%
%
\begin{equation}
A = %
\pmatrix{ 0 & i/24\vspace*{2pt}
\cr
i/24 & 0 }, \qquad B =
\pmatrix{ -1/48 & 0\vspace*{2pt}\cr 0 & 1/48 }.
\end{equation}
\end{proposition}

\begin{pf}
See Appendix~\ref{apriemann}.
\end{pf}
%

\begin{remark}
\label{reuniformity}
The\vspace*{1pt} error terms in (\ref{zkas}) are uniform in the parameters $\{
\alpha_{k}\}_{k=1}^{m-1}$ belonging to $\Omega$ [cf. (\ref{omset})],
$\{\tau_{k}\}_{k=1}^{m-1}$ belonging\vspace*{1pt} to a compact subset of $\mathbb
{R}$ and $x_{0}$ belonging to a compact subset of $(-1+\delta,1-\delta)$.
Furthermore,\vspace*{1pt} every such error term is an analytic function in the
variables $\{\alpha_{k}\}_{k=1}^{m-1}$ whose derivatives with respect
to $\alpha_{j}$ have the same order in $N$ and have the same
uniformity property described above. Hence, in the remainder of this
section it will be implicit that the error terms involved are of this form.
\end{remark}

Now inserting the above asymptotics (\ref{zkas}) into the differential
identity (\ref{diffid}), we obtain:
%

\begin{proposition}\label{prop}
Let $\varphi_{N}$ denote the characteristic function of the\break stochastic
process $W^{(\eta)}_{N}(\tau)$ defined in (\ref{multint}). Then in
the limit $N \to\infty$, we have
%
%
\begin{eqnarray}
\varphi_{N}(\alpha_{1},\ldots,
\alpha_{m-1}) &=& \exp\Biggl(N\sum_{k=1}^{m-1}
\alpha_{k} \bigl(\operatorname{Re}\bigl(g(z_{k}) \bigr
)-\operatorname{Re}
\bigl(g(z_{m}) \bigr) \bigr)\nonumber
\\
\label{charasymptotics}
&&{}+\sum_{k,j=1}^{m-1}\frac{\alpha_{k}\alpha_{j}}{2}
\bigl(\phi^{(\eta)}_{0}(\tau_{k}) +
\phi^{(\eta)}_{0}(\tau_{j})-\phi
^{(\eta)}_{0}(\tau_{k}-\tau_{j}) \bigr)
\\
\nonumber
&&{}+O \bigl(d_{N}^{-1} \bigr)+ O \biggl(N
\log(d_{N}) \exp\biggl(-c_{1}\frac{N}{d_{N}} \biggr)
\biggr) \Biggr),
\end{eqnarray}
where $g(z)$ is defined in (\ref{gfun}) and $\phi^{(\eta)}_{0}(\tau
)$ in (\ref{covphi}). The asympotics in (\ref{charasymptotics}) hold
uniformly in the same sense described in Remark~\ref{reuniformity}.
\end{proposition}


\begin{remark}
Notice that the asymptotics in (\ref{charasymptotics}) consist of both
\textit{global} error terms, which become large when $d_{N} \sim1$
and \textit{local} error terms, which become large when $d_{N} \sim
N$. Throughout the following proof, we will write $e_{N}$ for the local
error term of order
%
%
\begin{equation}
e_{N} = \log(d_{N})\exp\biggl(-c_{1}
\frac{N}{d_{N}} \biggr).
\end{equation}
\end{remark}

\begin{pf*}{Proof of Proposition~\protect\ref{prop}}
We remind the reader that the prime $'$ always denotes differentiation
with respect to $\alpha_{j}$. We begin by considering the second line
of (\ref{diffid}). Taking into account $\alpha_{m}=-(\alpha
_{1}+\cdots+\alpha_{m-1})$, we insert (\ref{zkas}) into (\ref{yas})
and make use of the explicit formula (\ref{pinf}) for $P_{\infty
}(z)$. Straightforward calculation then gives
%
%
\begin{eqnarray}
&& Y_{11}(\sqrt{2N}z_{k})'Y_{22}(
\sqrt{2N}z_{k})-Y_{21}(\sqrt{2N}z_{k})'Y_{12}(
\sqrt{2N}z_{k})\nonumber
\\
&&\qquad= \bigl(P_{\infty}(z_{k})
\bigr)'_{11} \bigl(P_{\infty}(z_{k})
\bigr)_{22}- \bigl(P_{\infty
}(z_{k})
\bigr)'_{21} \bigl(P_{\infty}(z_{k})
\bigr)_{12}
\nonumber
\\[-8pt]
\label{diffofpinfs}
\\[-8pt]
\nonumber
&&\qquad\quad{}+O \bigl(N^{-1} \bigr)+O(e_{N})
\\
\label{diffofys}
&&\qquad= C(z_{m},z_{k})-C(z_{j},z_{k})+O
\bigl(d_{N}^{-1} \bigr)+O(e_{N}),
\end{eqnarray}
where we introduced
%
%
\begin{equation}
\label{phifn}
C(\mu,z) = \frac{\sqrt{z+1}\sqrt{z-1}}{2\pi}\int_{-1}^{1}
\frac
{\log|x-\mu|}{\sqrt{1-x^{2}}}\frac{dx}{z-x},
\end{equation}
and (\ref{diffofys}) was obtained from (\ref{diffofpinfs}) using the
estimate $\mathcal{D}_{\infty}=1+O(d_{N}^{-1})$. Since
$C(z_{j},\overline{z_{k}}) = \overline{C(z_{j},z_{k})}$, we find from
(\ref{diffofys}) that
%
%
\begin{eqnarray}
\label{ysum} &&\hspace*{6pt}\qquad\frac{1}{2}\sum_{k=1}^{2m}
\alpha_{k} \bigl(Y_{11}(\sqrt{2N}z_{k})'Y_{22}(
\sqrt{2N}z_{k})-Y_{21}(\sqrt{2N}z_{k})'Y_{12}(
\sqrt{2N}z_{k}) \bigr)
\\
\label{ysum2} &&\hspace*{6pt}\qquad\qquad=\sum_{k=1}^{m}
\alpha_{k} \bigl(\operatorname{Re}\bigl(C(z_{m},z_{k})
\bigr)-\operatorname{Re}\bigl(C(z_{j},z_{k}) \bigr) \bigr)+O
\bigl(d_{N}^{-1} \bigr)+O(e_{N})
\\
\label{ysum3} &&\hspace*{6pt}\qquad\qquad=\sum_{k=1}^{m-1}
\alpha_{k} \bigl(\phi^{(\eta)}_{0}(\tau
_{k})+\phi^{(\eta)}_{0}(\tau_{j})-
\phi^{(\eta)}_{0}(\tau_{k}-\tau_{j})
\bigr)+O \bigl(d_{N}^{-1} \bigr)+O(e_{N}).
\end{eqnarray}
To obtain (\ref{ysum3}) from (\ref{ysum2}), we used the formula
(\ref{szasy2}) to compute the asymptotics of $\operatorname
{Re}(C(z_{j},z_{k}))$
and used that $\alpha_{m}=-(\alpha_{1}+\cdots+\alpha_{m-1})$.

Now let us compute the asymptotics of the coefficients $\beta_{N}$,
$\gamma_{N}$ and $\chi_{N-1}$ defined in (\ref{coeffid}) and
appearing in the first line of (\ref{diffid}). As usual, these
quantities are all obtained by expanding all $z$-dependent quantities
appearing in (\ref{yas}) in powers of $1/z$. First, the Szeg\"o
function (\ref{sz}) satisfies $\mathcal{D}(z) = \mathcal{D}_{\infty
}(1+\mathcal{D}_{1}/z+(\mathcal{D}_{1}^{2}/2+\mathcal
{D}_{2})/z^{2}+O(z^{-3}))$ as $z \to\infty$, where
%
%
\begin{eqnarray}
\mathcal{D}_{1} &=& -\frac{1}{2}\sum
_{k=1}^{m}\alpha_{k}\operatorname{Re}\biggl(
\frac{1}{z_{k}+\sqrt{z_{k}+1}\sqrt{z_{k}-1}} \biggr),
\nonumber
\\[-8pt]
\label{d2}
\\[-8pt]
\nonumber
\mathcal{D}_{2} &=& -\frac{1}{8}\sum
_{k=1}^{m}\alpha_{k}\operatorname{Re}\biggl(
\frac{1}{(z_{k}+\sqrt{z_{k}+1}\sqrt{z_{k}-1})^{2}} \biggr),
\end{eqnarray}
and second, use of the definitions (\ref{pinf}) and (\ref{gfun})
shows that for $z \to\infty$
%
%
\begin{equation}\qquad
g(z) = \log(z)-\frac{1}{8z^{2}}+O \bigl(z^{-4} \bigr), \qquad
a(z)=1- \frac
{1}{2z}+\frac{1}{8z^{2}}+O \bigl(z^{-3} \bigr).
\end{equation}

Then expanding (\ref{zkas}) at $z=\infty$, we can compare with (\ref
{coeffid}) and obtain
\begin{eqnarray*}
\beta_{N}&=& \sqrt{2N} \biggl(-\mathcal{D}_{1}+
\frac{A_{11}}{N}+O \biggl(\frac{1}{Nd_{N}} \biggr)+O(e_{N}) \biggr),
\\
\gamma_{N}&=& 2N \biggl(1/8-N/8+\mathcal{D}_{1}^{2}/2-
\mathcal{D}_{2}+\frac{B_{11}-A_{11}\mathcal{D}_{1}-iA_{12}/2}{N}\\
&&{}+O
\biggl(\frac{1}{Nd_{N}}
\biggr)+O(e_{N}) \biggr),
\\
\chi_{N-1}^{2}&=& \frac{2^{N-1}}{\sqrt{\pi}(N-1)!} \biggl(\frac
{1}{\mathcal{D}_{\infty}^{2}}+
\frac{1}{N} \biggl(\frac{1}{12\mathcal
{D}_{\infty}^{2}}+2iA_{21} \biggr)+O
\biggl( \frac{1}{Nd_{N}} \biggr)+O(e_{N}) \biggr).
\end{eqnarray*}
A similar computation shows that the asymptotics of $\chi^{2}_{N}$ are
given by
%
%
\begin{equation}\qquad
\chi^{2}_{N} = \frac{2^{N}}{\sqrt{\pi}N!} \biggl(\frac{1}{\tilde
{\mathcal{D}}^{2}_{\infty}}+
\frac{1}{N} \biggl(\frac{1}{12\tilde
{\mathcal{D}}^{2}_{\infty}}+2iA_{12} \biggr)+O
\biggl( \frac
{1}{Nd_{N}} \biggr)+O(e_{N}) \biggr),
\end{equation}
where $\tilde{\mathcal{D}}_{\infty}$ denotes the quantity (\ref
{dinf}) with rescaled singularities $\tilde{z}_{k} =\break  \sqrt
{2N/(2N+2)}z_{k}$. This rescaling is necessary when estimating $\chi
^{2}_{N}$, because without it one obtains asymptotics with respect to
the weight $w(x)=\prod_{j}|x-\sqrt{2N+2}z_{k}|^{\alpha_{k}}$.
Cumbersome though routine manipulations with the above asymptotics yield
%
%
\begin{eqnarray}
\qquad-N(\log\chi_{N}\chi_{N-1})^{\prime}&=& 2N
\bigl(C(z_{j},\infty)-C(z_{m},\infty) \bigr)+O
\bigl(d_{N}^{-1} \bigr)+O(Ne_{N}),
\nonumber
\\[-8pt]
\label{meanterm}
\\[-8pt]
\nonumber
2 \bigl(\gamma_{N}'-\beta_{N}
\beta_{N}' \bigr)&=& -4N\mathcal{D}_{2}'+O
\bigl(d_{N}^{-1} \bigr)+O(Ne_{N}),
\end{eqnarray}
and
%
%
\begin{eqnarray}
%
(\log\chi_{N}\chi_{N-1})'Y_{11}(
\sqrt{2N}z_{k})Y_{22}(\sqrt{2N}z_{k}) &=& O
\bigl(d_{N}^{-1} \bigr)+O(e_{N}),
\nonumber
\\[-8pt]
\label{smallterms}
\\[-8pt]
\nonumber
2 \biggl(\frac{\chi_{N-1}}{\chi_{N}} \biggr)^{2} \biggl(\log
\frac
{\chi_{N-1}}{\chi_{N}}
\biggr)^{\prime} &=&  O \bigl(d_{N}^{-1}
\bigr)+O(e_{N}),
\end{eqnarray}
where we introduced
%
%
\begin{eqnarray}
C(\mu,\infty) &=& \lim_{z \to\infty}C(\mu,z) =
\frac{1}{2\pi
}\int_{-1}^{1}\frac{\log|x-\mu|}{\sqrt{1-x^{2}}}
\,dx
\\
\label{phiinf}
&=& \frac{1}{2}\log|z+\sqrt{z+1}\sqrt{z-1}|-\frac{1}{2}\log(2).
\end{eqnarray}


Using the explicit formulae (\ref{phiinf}) and (\ref{d2}), we get
%
%
\begin{eqnarray}
&& 2 \bigl(C(z_{j},\infty)-C(z_{m},\infty) \bigr)-4
\mathcal{D}'_{2} = \operatorname{Re}\bigl(g(z_{j}) \bigr)-
\operatorname{Re}\bigl(g(z_{m}) \bigr),
\end{eqnarray}
where we exploited the convenient identity (see, e.g., the
derivation of equation~(7.89) in \cite{De99})
%
%
\begin{equation}\qquad
\log|z+\sqrt{z+1}\sqrt{z-1}|+\frac{1}{2}\operatorname{Re}\biggl
(\frac
{1}{(z+\sqrt{z+1}\sqrt{z-1})^{2}}
\biggr) = \operatorname{Re}\bigl(g(z) \bigr).
\end{equation}
Now inserting (\ref{meanterm}), (\ref{ysum3}) and (\ref{smallterms})
into (\ref{diffid}), we obtain
%
%
\begin{eqnarray}
 %
&&\frac{\partial}{\partial\alpha_{j}}\log
\varphi_{N}(\alpha_{1},\ldots,\alpha_{m-1})\hspace*{-6pt}\nonumber\\
\label{logasy}
 &&\qquad = N
\bigl(\operatorname{Re}\bigl(g(z_{j}) \bigr)-\operatorname
{Re}\bigl(g(z_{m}) \bigr)
\bigr)\hspace*{-6pt}
\\
\nonumber
&&\qquad\quad{}+\sum_{k=1}^{m-1}\alpha_{k}
\bigl(\phi^{(\eta)}_{0}(\tau_{k})+
\phi^{(\eta)}_{0}(\tau_{j})-\phi^{(\eta)}_{0}(
\tau_{k}-\tau_{j}) \bigr)+O \bigl(d_{N}^{-1}
\bigr)+O(Ne_{N}).\hspace*{-6pt} 
\end{eqnarray}
Note that the error terms in (\ref{logasy}) hold \textit{uniformly}
in the parameters $(\alpha_{k})_{k=1}^{m-1}$ (see Remark~\ref
{reuniformity}), so that we may integrate both sides of (\ref{logasy})
according to the procedure discussed in Section $5$ of \cite
{K07}, arriving at the asymptotics (\ref{charasymptotics}).
\end{pf*}
%
%
\begin{pf*}{Proof of Theorems \protect\ref{thmaintheorem} and
\protect\ref{thcompactconv}}
Bearing\vspace*{1pt} in mind Remark~\ref{reuniformity}, we differentiate (\ref
{charasymptotics}) with respect to the parameters $(\alpha
_{k})_{k=1}^{m-1}$ and evaluate near the origin, leading to
%
%
\begin{eqnarray}\label{mean}
&&\hspace*{6pt}\mathbb{E} \bigl\{ W^{(\eta)}_{N}(\tau) \bigr\} = N \bigl
(\operatorname{Re}
\bigl(g(z_{k}) \bigr)-\operatorname{Re}\bigl(g(z_{m}) \bigr) \bigr)+O
\bigl(d_{N}^{-1} \bigr) + O(Ne_{N}),
\\
&& \hspace*{6pt}\operatorname{Cov} \bigl\{W^{(\eta)}_{N}(\tau), W^{(\eta)}_{N}(
\upsilon) \bigr\}
\nonumber
\\[-8pt]
\label{covar}
\\[-8pt]
\nonumber
&&\hspace*{6pt}\qquad = \phi^{(\eta)}_{0}(\tau)+
\phi^{(\eta)}_{0}(\upsilon)-\phi^{(\eta
)}_{0}(
\tau-\upsilon)+O \bigl(d_{N}^{-1} \bigr) +O(Ne_{N}),
\end{eqnarray}
where\vspace*{1pt} the error terms are uniform in $\tau$ and $\upsilon$ varying in
a compact subset of $\mathbb{R}$. Then defining the centered process
$\tilde{W}^{(\eta)}_{N}(\tau) = W^{(\eta)}_{N}(\tau)-\mathbb{E}\{
W^{(\eta)}_{N}(\tau)\}$ we immediately find from (\ref{mean}) and
(\ref{charasymptotics}) that in the mesoscopic regime (\ref{M2}), we have
%
%
\begin{eqnarray}
&& \lim_{N \to\infty}\mathbb{E} \bigl\{e^{i\sum_{k=1}^{m}s_{k}\tilde
{W}^{(\eta)}_{N}(\tau_{k})} \bigr\}
\nonumber
\\[-8pt]
\\[-8pt]
\nonumber
&&\qquad= \exp
\Biggl(-\frac
{1}{2}\sum_{k=1}^{m}\sum
_{j=1}^{m}s_{k}s_{j}
\bigl(\phi_{0}(\tau_{k})+\phi_{0}(
\tau_{j})-\phi_{0}(\tau_{k}-
\tau_{j}) \bigr) \Biggr),
\end{eqnarray}
%
where $(s_{k})_{k=1}^{m} \in\mathbb{R}^{m}$. Theorem~\ref
{thmaintheorem} follows immediately. To complete the proof of
Theorem~\ref{thcompactconv}, it suffices to note that the error terms in
(\ref{covar}) are uniform, so that the sequence $(\mathbb{E}\{(\tilde
{W}_{N}(\tau))^{2}\})_{N=1}^{\infty}$ is uniformly bounded.
\end{pf*}

\section{Convergence to white noise in the spectral representation}
\label{sefouriercoeff}
The main achievement of the previous section was to prove that for any
mesoscopic scales of the form (\ref{M2}), the process $\tilde
{W}^{(\eta)}_{N}(\tau)$ converges in the sense of finite-dimensional
distributions to the regularized fractional Brownian motion $B^{(\eta
)}_{0}(\tau)$. We also proved Theorem~\ref{thcompactconv} which
extends this convergence to an appropriate function space.

In this section, we will study $\tilde{W}^{(\eta)}_{N}(\tau)$ from a
different point of view, namely by means of the Fourier coefficients
$b_{N}(s)$ appearing in the spectral decomposition~(\ref
{detforident}). We remind the reader of the definition
%
%
\begin{equation}\label{bee2}
b_{N}(s) = \frac{1}{\sqrt{s}}\operatorname{Tr} \bigl(e^{-isd_{N}(\mathcal
{H}-x_{0}I)}
\bigr), \qquad s>0.
\end{equation}
A useful and interesting feature of the integral representations (\ref
{detforident}) and its \mbox{$N \to\infty$} limit (\ref{intr3}) is that they
are suggestive of a corresponding limiting law satisfied by the
coefficients $b_{N}(s)$. Namely,\vspace*{1pt} we expect that $b_{N}(s)$ should
``converge'' to the white noise measure $B_{c}(ds)/\sqrt{2}$. The
precise mode of the convergence we consider is described in
Theorem~\ref{thfourierconv} and it is our goal in this section to
prove this result.


By its very definition, the white noise measure $B_{c}(ds)$ cannot be
understood in a pointwise sense and must be regularized by integrating
against a test function. We will consider test functions $\xi\in
C^{\infty}_{0}(\mathbb{R}_{+})$, that is, $\xi$ is a smooth function
with compact support on $\mathbb{R}_{+}$. Then we have the correspondence:
%
%
\begin{equation}\label{cn2}
c_{N}(\xi)=\int_{0}^{\infty}
\xi(s)b_{N}(s)\,ds = \sum_{j=1}^{N}f
\bigl(d_{N}(x_{j}-x_{0}) \bigr) =:
X_{N}(f),
\end{equation}
where
%
%
\begin{equation}
\label{gxi} f(x) = \int_{0}^{\infty}
\frac{\xi(s)}{\sqrt{s}}e^{-isx}\,ds.
\end{equation}
By our assumptions on $\xi$, it follows that $f$ belongs to the
Schwartz space of rapidly decaying smooth functions, that is, $f \in
S(\mathbb{R})$ where
%
%
\begin{equation}\label{schwartz}
S(\mathbb{R}) = \biggl\{f \in C^{\infty}(\mathbb{R}):  \sup
_{x
\in\mathbb{R}} \biggl|x^{\gamma}\frac{d^{\beta}f(x)}{dx^{\beta
}} \biggr|<\infty,
\gamma,\beta=0,1,2,\ldots\biggr\}.
\end{equation}
In the following three subsections, we will obtain results for the
mean, variance and distribution of the random variable (\ref{cn2}) as
$N \to\infty$.

\subsection{Mean}
We begin by proving that centering is not required in Theorem~\ref
{thfourierconv}.
%

\begin{proposition}
On any mesoscopic scales of the form $d_{N} = N^{\alpha}$ with any
$\alpha\in(0,1)$, we have
%
%
\begin{equation}
\mathbb{E} \bigl\{c_{N}(\xi) \bigr\} = O \bigl(d_{N}^{-1}
\bigr), \qquad N \to\infty.
\end{equation}
\end{proposition}

\begin{pf}
We write the expectation above as an integral over the normalized
density of states $\rho_{N}(x)$,
%
%
\begin{equation}
\label{fullint}
\mathbb{E} \bigl\{c_{N}(\xi) \bigr\}= N\int
_{-\infty}^{\infty
}f \bigl(d_{N}(x-x_{0})
\bigr)\rho_{N}(x)\,dx,
\end{equation}
where
%
%
\begin{equation}
\rho_{N}(x) = \frac{1}{N}\mathbb{E} \Biggl\{\sum
_{j=1}^{N}\delta(x-x_{j}) \Biggr\}.
\label{dos}
\end{equation}
Firstly, note that the tails of the integral (\ref{fullint}) can be
removed using the rapid decay of $f$. For any $\varepsilon>0$, we have
%
%
\begin{equation}
\label{inst}
\mathbb{E} \bigl\{c_{N}(\xi) \bigr\}= N\int
_{x_{0}-\varepsilon}^{x_{0}+\varepsilon
}f \bigl(d_{N}(x-x_{0})
\bigr)\rho_{N}(x)\,dx + O \bigl(Nd_{N}^{-\infty} \bigr),
\end{equation}
where here and elsewhere, the notation $O(Nd_{N}^{-\infty})$ refers to
a quantity that is $O(Nd_{N}^{-\gamma})$ for any $\gamma>0$. Such a
contribution tends to zero for the power law scales $d_{N} = N^{\alpha
}$ with any $\alpha\in(0,1)$. Then for small enough $\varepsilon$, we
have the uniform estimate (see \cite{PS11}, Chapter~5.2)
%
%
\begin{equation}
\label{denslim}
\rho_{N}(x) = \frac{2}{\pi}\sqrt{1-x^{2}}+O
\bigl(N^{-1} \bigr), \qquad x \in(x_{0}-
\varepsilon,x_{0}+\varepsilon).
\end{equation}

After inserting (\ref{denslim}) into (\ref{inst}), we find that
%
%
\begin{equation}
\label{climit}\qquad
\mathbb{E} \bigl\{c_{N}(\xi) \bigr\} =
\frac{2N}{\pi} \int_{x_{0}-\varepsilon
}^{x_{0}+\varepsilon}f
\bigl(d_{N} (x-x_{0}) \bigr)\sqrt{1-x^{2}}\,
dx+E_{N}+O \bigl(Nd_{N}^{-\infty} \bigr),\hspace*{-12pt}
\end{equation}
where the error term $E_{N} = O(d_{N}^{-1})$, since
%
%
\begin{eqnarray}
|E_{N}| &\leq &  C \Biggl|\int_{x_{0}-\varepsilon}^{x_{0}+\varepsilon
}f
\bigl(d_{N}(x-x_{0}) \bigr)\,dx \Biggr|\leq\frac{C}{d_{N}}
\int_{-\infty
}^{\infty}\bigl|f(x)\bigr|\,dx.
\end{eqnarray}

Similarly, we can replace the integration limits in (\ref{climit})
with $\pm1$ using the Schwartz property of $f$. We have
%
%
\begin{equation}\label{semiasympt}
\mathbb{E} \bigl\{c_{N}(\xi) \bigr\} = \frac{2N}{\pi}\int
_{-1}^{1}f \bigl(d_{N}(x-x_{0})
\bigr)\sqrt{1-x^{2}}\,dx + O \bigl(d_{N}^{-1}
\bigr).
\end{equation}

Next, we substitute $f$ with the definition (\ref{gxi}) and
interchange the order of integration [justified by the rapid decay of
$\xi(s)$] so that
%
%
\begin{eqnarray}
\qquad\mathbb{E} \bigl\{c_{N}(\xi) \bigr\} &= & \frac{2N}{\pi}\int
_{0}^{\infty}\xi(s)s^{-1/2}e^{isd_{N}x_{0}}
\int_{-1}^{1}e^{-isd_{N}x}
\sqrt{1-x^{2}} \, dx\,ds+O \bigl(d_{N}^{-1} \bigr)
\nonumber
\\[-8pt]
\label{bessel}
\\[-8pt]
\nonumber
&=&  2N\int_{0}^{\infty}\xi(s)s^{-3/2}J_{1}(d_{N}s)e^{isd_{N}x_{0}}
\, ds + O \bigl(d_{N}^{-1} \bigr),
\end{eqnarray}
where $J_{1}(z)$ is the Bessel function of index $1$. To complete the
proof, note that $J_{1}(d_{N}s)$ has an asymptotic expansion (for any
fixed $\gamma\in\mathbb{N}$ and $s>0$) as $N \to\infty$,
%
%
\begin{eqnarray}
\sqrt{\frac{\pi}{2}}J_{1}(d_{N}s) &=&
\cos(d_{N}s-3\pi/4)\sum_{k=0}^{\gamma-1}
\frac{C_{k}}{d_{N}^{2k+1/2}s^{2k+1/2}}
\nonumber
\\[-8pt]
\label{sumerrors}
\\[-8pt]
\nonumber
&&{}+\sin(d_{N}s-3\pi/4)\sum_{k=0}^{\gamma-1}
\frac
{D_{k}}{d_{N}^{2k+3/2}s^{2k+3/2}}+E_{N}(s),
\end{eqnarray}
where the error term satisfies the bound $|E_{N}(s)| \leq|C_{\gamma
}d_{N}^{-2\gamma-1/2}s^{-2\gamma-1/2}|$ and $C_{k},D_{k}$ are
constants depending only on $k$. Such asymptotics can be found in, for
example, \cite{NIST} or \cite{K11}.

Inserting (\ref{sumerrors}) into (\ref{bessel}), we see that the
contribution from each term in the sum in (\ref{sumerrors}) is an
oscillatory integral of order $O(Nd_{N}^{-\infty})$, as follows from
repeated integration by parts. The final\vspace*{1pt} error term $E_{N}(s)$ is
integrable with respect to $\xi(s)$ and gives rise to an error of
order $O(Nd_{N}^{-2\gamma})$. Since $\gamma>0$ was arbitrary, we
conclude that the term proportional to $N$ in (\ref{semiasympt}) is in
fact asymptotically smaller than the error term. This completes the
proof of the proposition.
\end{pf}


\subsection{Covariance}
\label{secov}
Having studied the expectation of $b_{N}(s)$ in the previous
subsection, we now consider the fluctuations. In the \hyperref[sec1]{Introduction}, it
was remarked, in accordance with the expected white noise limit for
$b_{N}(s)$ that we should have $\lim_{N \to\infty}\mathbb{E}\{
b_{N}(s_{1})\overline{b_{N}(s_{2})}\} = \delta(s_{1}-s_{2})$. In this
subsection, we will make this assertion precise by proving that
%
%
\begin{equation}
\label{deltacors} \lim_{N \to\infty}\mathbb{E} \bigl\{c_{N}(
\xi_{1})\overline{c_{N}(\xi_{2})} \bigr\} = \int
_{0}^{\infty}\xi_{1}(s)\overline{
\xi_{2}(s)}\,ds
\end{equation}
for all smooth functions $\xi_{1},\xi_{2}$ with compact support on
$\mathbb{R}_{+}$.

It turns out that there is an exact finite-$N$ formula for the
covariance (see equation (4.2.38) in \cite{PS11}):
%
%
\begin{equation}\label{orig}\qquad
\mathbb{E} \bigl\{\tilde{X}_{N}(f_{1})\tilde{X}_{N}(f_{2})
\bigr\}=\frac
{1}{8}\int_{\mathbb{R}^{2}}\Delta
f_{1}(d_{N}x)\Delta f_{2}(d_{N}x)K_{N}^{2}(x_{1},x_{2})
\,dx_{1}\,dx_{2},
\end{equation}
where $f_{1}$ and $f_{2}$ are defined in terms of $\xi_{1}$ and $\xi
_{2}$ as in formula (\ref{gxi}) and we introduced the notation $\Delta
f(x) = f(x_{1})-f(x_{2})$ for any $f$. The function
$K_{N}(x_{1},x_{2})$ is the kernel of the GUE ensemble (see, e.g., \cite{Meh04,PS11}) having the explicit formula
%
%
\begin{equation}\label{guekernel}
K_{N}(x,y) = \frac{\psi^{(N)}_{N}(x_{1})\psi^{(N)}_{N-1}(x_{2})-\psi
^{(N)}_{N}(x_{2})\psi^{(N)}_{N-1}(x_{1})}{x_{1}-x_{2}},
\end{equation}
where
%
%
\begin{equation}
\label{wavefns}
\psi^{(N)}_{l}(x) = e^{-Nx^{2}}P^{(N)}_{l}(x),
\end{equation}
and $P^{(N)}_{l}(x)$ are (rescaled) Hermite polynomials, normalized by
the condition that $\{\psi^{(N)}_{l}\}_{l=1}^{\infty}$ forms an
orthonormal family on $\mathbb{R}$. By making use of the known
Plancherel--Rotach asymptotics for the functions $\psi^{(N)}_{l}(x)$,
we deduce the following covariance formula. After noting the
correspondence (\ref{gxi}), we immediately derive from it the $\delta$-correlations (\ref{deltacors}).

%
\begin{proposition}\label{prop1}
Let the test functions $f_{1}$ and $f_{2}$ belong to the Schwartz space
$S(\mathbb{R})$ defined in (\ref{schwartz}) and consider the
mesoscopic regime $d_{N} = N^{\alpha}$ with any $\alpha\in(0,1)$. We have
%
%
\begin{equation}\label{eqnfour}
\lim_{N \to\infty}\mathbb{E} \bigl\{\tilde{X}_{N}(f_{1})
\tilde{X}_{N}(f_{2}) \bigr\} = \frac{1}{2\pi}\int
_{-\infty}^{\infty
}|s|\hat{f_{1}}(s)
\hat{f_{2}}(-s)\,ds,
\end{equation}
where $\hat{f}(s) = (2\pi)^{-1/2}\int_{-\infty}^{\infty
}f(x)e^{-isx}\,dx$.
\end{proposition}

%
\begin{remark}
Formula (\ref{eqnfour}) is already known for $C^{1}$ functions with
compact support, as in Theorem~5.2.7(iii) of \cite{PS11}. It was also
proved recently in \cite{EK13a} for a class of Wigner matrices with
$f$ a Schwartz test function, but only up to scales $d_{N} = N^{\alpha
}$ with any $0<\alpha<1/3$. Our main contribution in this subsection
is to adapt the argument given in \cite{PS11} to our test functions
$f$ in (\ref{gxi}), which cannot be compactly supported due to our
assumptions on $\xi$. We note that our proof holds on the full range
$0 < \alpha< 1$ and that the smoothness hypothesis can be relaxed to
$C^{1}$ functions with rapid decay at $\pm\infty$.
\end{remark}
\begin{pf*}{Proof of Proposition~\protect\ref{prop1}}
Here, we only consider the contribution to integral (\ref{orig})
coming from the square $I_{\delta}^{2} = [-(1-\delta),(1-\delta
)]^{2}$ for some small $\delta>0$. In Appendix~\ref{apintcomp}, we
will show that the complement of this region can be neglected for small
enough $\delta$. We will need the following asymptotic formula for the
functions $\psi^{(N)}_{N+k}$ defined in (\ref{wavefns}). Uniformly
for $|x|<(1-\delta)$ and $k=O(1)$, we have
%
%
\begin{eqnarray}
\qquad\psi^{(N)}_{N+k}(x) &=& \biggl(
\frac{2}{\pi\sqrt{1-x^{2}}} \biggr)^{1/2}\cos\bigl(N\alpha(x)+(k+1/2)
\cos^{-1}(x)-\pi/4 \bigr)
\nonumber
\\[-8pt]
\label{planch2}
\\[-8pt]
\nonumber
&&{}+O \bigl(N^{-1} \bigr),
\end{eqnarray}
where $\alpha(x) = 2\int_{-1}^{x}dt \sqrt{1-t^{2}}$. Formula
(\ref{planch2}) follows immediately from the classical asymptotic
results of Plancherel and Rotach (see Sections $5$ in \cite{PS11} and
$8$ in~\cite{Sze39}).

Now, using the symmetry about the line $x_{1}=x_{2}$, we see that the
integral (\ref{orig}) restricted to $I_{\delta}^{2}$ can be written
in the convenient form,
%
%
\begin{equation}
\label{intsquare} \frac{1}{4}\int_{I_{\delta}^{2}}\frac{\Delta
f_{1}(d_{N}x)}{\Delta
x}
\frac{\Delta f_{2}(d_{N}x)}{\Delta x}\mathcal{F}_{N}(x_{1},x_{2})\,
dx_{1}\,dx_{2},
\end{equation}
where
%
%
\begin{eqnarray}
\mathcal{F}_{N}(x_{1},x_{2})& =&
\psi^{(N)}_{N}(x_{1})^{2}\psi
^{(N)}_{N-1}(x_{2})^{2}
\nonumber
\\[-8pt]
\label{cfn}
\\[-8pt]
\nonumber
&&{}-
\psi^{(N)}_{N}(x_{1})\psi^{(N)}_{N-1}(x_{1})
\psi^{(N)}_{N}(x_{2})\psi^{(N)}_{N-1}(x_{2}).
\end{eqnarray}

We insert the Plancherel--Rotach formula (\ref{planch2}) into (\ref
{intsquare}) and denote $\theta(x) = \cos^{-1}(x)$. Using the double
angle formula for the cosine, we find that the contribution of $(\ref
{planch2})$ to the product of squares in (\ref{cfn}) is
%
%
\begin{eqnarray}\label{costerms1}
&&\qquad\frac{1+\cos(2N\alpha(x_{1})+\theta(x_{1})/2-\pi/4)+\cos
(2N\alpha(x_{2})-\theta(x_{2})/2-\pi/4)}{\pi^{2}\sqrt
{1-x_{2}^{2}}\sqrt{1-x_{1}^{2}}}\hspace*{-12pt}
\\
&&\qquad\qquad{}+\frac{\cos(2N\alpha(x_{1})+\theta(x_{1})/2-\pi/4)\cos(2N\alpha
(x_{2})-\theta(x_{2})/2-\pi/4)}{\pi^{2}\sqrt{1-x_{2}^{2}}\sqrt
{1-x_{1}^{2}}}\hspace*{-12pt}
\nonumber
\\[-8pt]
\label{costerms2}
\\[-8pt]
\nonumber
&&\qquad\qquad{}+O \bigl(N^{-1} \bigr).\hspace*{-12pt}
\end{eqnarray}
Inserting the oscillatory terms in lines (\ref{costerms1}) and (\ref
{costerms2}) into (\ref{intsquare}) gives rise to error terms that are
$O((N/d_{N})^{-\infty})$ as $N \to\infty$ for every $\delta>0$.
This can be shown by repeated integration by parts, using the fact that
$\alpha(x)$ is smooth and increasing on the interval $I_{\delta}$.
Combined with a similar calculation applied to the second term in
(\ref{cfn}), we see that the integral (\ref{intsquare}) is equal to
%
%
\begin{eqnarray}
&& \frac{1}{4\pi^{2}}\int_{I_{\delta}^{2}} \frac{\Delta
f_{1}(d_{N}x)}{\Delta x}
\frac{\Delta f_{2}(d_{N}x)}{\Delta x}\frac
{1-x_{1}x_{2}}{\sqrt{1-x_{1}^{2}}\sqrt{1-x_{2}^{2}}}\,dx_{1}\, dx_{2}+O
\bigl((N/d_{N})^{-\infty} \bigr)
\nonumber
\\
\label{underint}
&&\qquad =\frac{1}{4\pi^{2}}\int_{\mathbb{R}^{2}}\frac{\Delta
f_{1}(x)}{\Delta x}
\frac{\Delta f_{2}(x)}{\Delta x}\frac
{1-x_{1}x_{2}/d_{N}^{2}}{\sqrt{1-x_{1}^{2}/d_{N}^{2}}\sqrt
{1-x_{2}^{2}/d_{N}^{2}}}
\\
\nonumber
&&\qquad\quad{}\times\chi_{I_{N}}(x_{1})
\chi_{I_{N}}(x_{2})\, dx_{1}\,dx_{2}+O
\bigl((N/d_{N})^{-\infty} \bigr),
\end{eqnarray}
where $\chi_{I_{N}}(x_{1})$ is the indicator function on the set
$I_{N}=(-(1-\delta)d_{N},(1-\delta)d_{N})$.

Now Lebesgue's dominated convergence theorem can be applied to take the
limit under the integral in (\ref{underint}). Indeed, it is easy to
see that the integrand in~(\ref{underint}) is bounded by the
integrable function
%
%
\begin{equation}
\biggl(\frac{2}{\delta^{2}}-1 \biggr) \biggl|\frac{\Delta
f_{1}(x)}{\Delta x} \biggr| \biggl|
\frac{\Delta f_{2}(x)}{\Delta
x} \biggr|
\end{equation}
for any $N \in\mathbb{N}$, $(x_{1},x_{2})\in\mathbb{R}^{2}$ and
$0<\delta<1$. We finally see that for all $0<\delta<1$, we have
%
%
\begin{eqnarray}
&&\lim_{N \to\infty}\frac{1}{4}\int_{I_{\delta}^{2}}
\frac{\Delta
f_{1}(d_{N} x)}{\Delta x}\frac{\Delta f_{2}(d_{N} x)}{\Delta
x}\mathcal{F}_{N}(x_{1},x_{2})
\,dx_{1}\,dx_{2}
\nonumber
\\[-8pt]
\label{prefour}
\\[-8pt]
\nonumber
&&\qquad= \frac{1}{4\pi
^{2}}\int
_{\mathbb{R}^{2}}\frac{\Delta f_{1}(x)}{\Delta x}\frac
{\Delta f_{2}(x)}{\Delta x}
\,dx_{1}\,dx_{2}.
\end{eqnarray}
Rewriting $f_{1}$ and $f_{2}$ in terms of their Fourier transforms and
applying the Plancherel theorem gives the identity
%
%
\begin{eqnarray}
&&\frac{1}{4\pi^{2}}\int_{\mathbb{R}^{2}}\frac
{f_{1}(x_{1})-f_{1}(x_{2})}{x_{1}-x_{2}}
\frac
{f_{2}(x_{1})-f_{2}(x_{2})}{x_{1}-x_{2}}\,dx_{1}\,dx_{2}
\nonumber
\\[-8pt]
\\[-8pt]
\nonumber
&&\qquad= \frac
{1}{2\pi}\int
_{\mathbb{R}}|s|\hat{f_{1}}(s)\hat{f_{2}}(-s)
\,ds,
\end{eqnarray}
which is precisely the RHS of (\ref{eqnfour}). To complete the proof,
we just need to show that the integral (\ref{orig}) restricted to the
complement of the square $I_{\delta}^{2}$ can be neglected in the
limit $N \to\infty$. Namely, we prove in the \hyperref[app]{Appendix} that
%
%
\begin{equation}
\label{432432}\qquad
\lim_{N \to\infty}\int_{(I_{\delta}^{2})^{\mathrm{c}}}
\Delta f_{1}(d_{N}x)\Delta f_{2}(d_{N}x)K_{N}^{2}(x_{1},x_{2})
\,dx_{1}\, dx_{2}=O(\delta), \qquad\delta\to0,\hspace*{-12pt}
\end{equation}
and so complete the proof of the proposition by choosing $\delta>0$
sufficiently small.
\end{pf*}
%
\subsection{Convergence in distribution}
\label{sefulldist}
The aim of this subsection is to study the full distribution of the
coefficients $b_{N}(s)$ and ultimately to prove Theorem~\ref
{thfourierconv}. First, we need a preliminary result regarding the
stochastic process $\tilde{W}^{(\eta)}_{N}(\tau)$. It will be
convenient to consider the \textit{increments}
%
%
\begin{eqnarray}
%
&&\Delta_{p} \bigl[\tilde{W}_{N}^{(\eta)}
\bigr](\tau)\nonumber
\\
\label{detforident2}
&&\qquad:= \tilde{W}_{N}^{(\eta
)}(\tau)-
\tilde{W}_{N}^{(\eta)}(\tau+p)
\\
\nonumber
&&\qquad=  \frac{1}{2}\int_{0}^{\infty}
\frac{e^{-\eta s}}{\sqrt{s}} \bigl\{ \bigl[1-e^{-ips} \bigr
]e^{-i\tau s}
\tilde{b}_{N}(s)+ \bigl[1-e^{ips} \bigr]e^{i\tau
s}
\overline{\tilde{b}_{N}(s)} \bigr\}\,ds,
%
\end{eqnarray}
where $\tilde{b}_{N}(s) = b_{N}(s)-\mathbb{E}\{b_{N}(s)\}$.

Similarly, the corresponding limiting object is given by the following
stationary Gaussian process:
%
%
\begin{eqnarray}
%
&& \Delta_{p} \bigl[B^{(\eta)}_{0}
\bigr](\tau)\nonumber\\
\label{b0inc}
 &&\qquad:= B^{(\eta)}_{0}(\tau)-B^{(\eta
)}_{0}(
\tau+p)
\\
\nonumber
&&\qquad= \frac{1}{2\sqrt{2}}\int_{0}^{\infty}
\frac{e^{-\eta s}}{\sqrt
{s}} \bigl\{ \bigl[1-e^{-ips} \bigr]e^{-i\tau s}B_{c}(ds)+
\bigl[1-e^{ips} \bigr]e^{i\tau
s}\overline{B_{c}(ds)}
\bigr\}.
\end{eqnarray}
%

\begin{proposition}
\label{proptailsh}
Let $p \in\mathbb{R}$. For any $h \in S(\mathbb{R})$ and on any
power law scales $d_{N} = N^{\alpha}$ with $\alpha\in(0,1)$, we have
the convergence in distribution:
%
%
\begin{equation}\label{convhschwartz}
\qquad\int_{-\infty}^{\infty}h(\tau)\Delta_{p} \bigl[
\tilde{W}^{(\eta
)}_{N} \bigr](\tau)\,d\tau\stackrel{d} {
\Longrightarrow} \int_{-\infty
}^{\infty}h(\tau)
\Delta_{p} \bigl[B^{(\eta)}_{0} \bigr](\tau)\,d\tau,
\qquad N \to\infty.\hspace*{-12pt}
\end{equation}
\end{proposition}

\begin{pf}
The proof will be analogous to our proof of Theorem~\ref{thcompactconv}, the main difference being we must have good enough
control of the tails in the above integrals. This will be taken care of
by the rapid decay of $h$. To proceed, we fix some (arbitrary) $M \in
\mathbb{R}$ and $\delta_{0}>0$ and decompose the LHS of (\ref
{convhschwartz}) as
%
%
\begin{eqnarray}
%
&& \int_{-M}^{M}h(\tau)
\Delta_{p} \bigl[\tilde{W}^{(\eta)}_{N} \bigr](\tau)\,
d \tau+\int_{|\tau|\in[M,\delta_{0}d_{N}]}h(\tau)\Delta_{p} \bigl[
\tilde{W}^{(\eta)}_{N} \bigr](\tau)\,d\tau
\nonumber
\\[-8pt]
\label{decomp}
\\[-8pt]
\nonumber
&&\quad {}+\int
_{|\tau| \in
[\delta_{0}d_{N},\infty)}h(\tau)\Delta_{p} \bigl[
\tilde{W}^{(\eta
)}_{N} \bigr](\tau)\,d\tau\hspace*{-12pt}
%
\end{eqnarray}
and label each of the integrals in (\ref{decomp}) with $\mathcal
{I}_{1}, \mathcal{I}_{2}$ and $\mathcal{I}_{3}$. Let us begin with
the first integral, $\mathcal{I}_{1}$. By Theorem~\ref
{thmaintheorem} and the Cram\'er--Wold device, the finite-dimensional
distributions of $\Delta_{p}[\tilde{W}^{(\eta)}_{N}](\tau)$
converge in law to those of $\Delta_{p}[B^{(\eta)}_{0}](\tau)$.
Furthermore, by the uniform estimate (\ref{covar}) we have that there
is a constant $C>0$ such that $\mathbb{E}\{(\Delta_{p}[B^{(\eta
)}_{0}(\tau)])^{2}\}\leq C$ for all $\tau\in[-M,M]$ and for all $N$.
Therefore, the hypotheses of Theorem~3 in \cite{G76} are satisfied and
we conclude that the first integral in (\ref{decomp}) converges in
distribution to the RHS of (\ref{convhschwartz}) in the limit $N \to
\infty$ followed by $M \to\infty$. To complete the proof, it
suffices to show that the second and third integrals in (\ref{decomp})
converge in probability to $0$ in the same limit.

For notational convenience, we just consider the contributions to
$\mathcal{I}_{2}$ and $\mathcal{I}_{3}$ where $\tau>0$ as the
situation $\tau<0$ is almost identical. By Chebyshev's inequality and
Cauchy--Schwarz, we have
%
%
\begin{equation}\label{varrhp}\qquad
\mathbb{P}\bigl\{|\mathcal{I}_{2}|>\varepsilon\bigr\} \leq\varepsilon^{-2}
\int_{M}^{\delta_{0}d_{N}}\bigl|h(\tau)\bigr|\,d\tau\int
_{M}^{\delta
_{0}d_{N}}\bigl|h(\tau)\bigr|\mathbb{E} \bigl\{
\Delta_{p} \bigl[\tilde{W}^{(\eta
)}_{N} \bigr](
\tau)^{2} \bigr\} \, d\tau.
\end{equation}
We will now argue that the variance term in (\ref{varrhp}) is
uniformly bounded. Since $|\tau|\leq\delta_{0}d_{N}$, by choosing
$\delta_{0}$ small enough we see that $|x_{0}+\tau/d_{N}| < 1-\delta
$ for some $\delta>0$ independent of $N$. Hence, the singularities of
the logarithm in (\ref{wntau}) remain inside the bulk region
$(-1+\delta,1-\delta)$ for all $N$ and we may apply the methods of
Section~\ref{semesoproof} with $m=2$ and weight [cf. (\ref{omdef})]
%
%
\begin{eqnarray}
\omega(z)&=& \biggl[\frac{(z-x_{0}(\tau,N)-p/d_{N})^{2}+(\eta
/d_{N})^{2}}{(z-x_{0}(\tau,N))^{2}+(\eta/d_{N})^{2}} \biggr]^{\alpha
/2},
\nonumber
\\[-8pt]
\\[-8pt]
\eqntext{\displaystyle x_{0}(\tau,N)=x_{0}+\tau/d_{N}.}
\end{eqnarray}
The only difference in the analysis of the Riemann--Hilbert problem
with this weight is that the new reference point $x_{0}(\tau,N)$ can
vary with $N$ in the small fixed neighbourhood $[x_{0}-\delta
_{0},x_{0}+\delta_{0}]$. However, all the estimates we obtain are
uniform for $x_{0}$ varying in compact subsets of $(-1+\delta,1-\delta
)$ so that the variance bound~(\ref{covar}) (with $\upsilon= \tau$)
remains valid. This implies that for some $N$-independent $C>0$,
%
%
\begin{equation}
\mathbb{P}\bigl\{|\mathcal{I}_{2}|>\varepsilon\bigr\} \leq\varepsilon^{-2}C
\biggl(\int_{M}^{\delta_{0}d_{N}}\bigl|h(\tau)\bigr|\,d\tau
\biggr)^{2} \to0,
\end{equation}
in the limit $N \to\infty$ followed by $M \to\infty$.

To bound the integral $\mathcal{I}_{3}$, we again apply Chebyshev's
inequality and exploit the rapid decay of $h$. We have
%
%
\begin{eqnarray}
&&\qquad\mathbb{P}\bigl\{|\mathcal{I}_{3}|>\varepsilon\bigr\}
\nonumber
\\[-8pt]
\\[-8pt]
\nonumber
&&\qquad\qquad\leq  \varepsilon^{-2}
\int_{\delta_{0}d_{N}}^{\infty}\int_{\delta_{0} d_{N}}^{\infty}
\mathbb{E} \bigl\{ h(\tau_{1})\Delta_{p} \bigl[
\tilde{W}^{(\eta)}_{N} \bigr](\tau_{1})\overline{h(
\tau_{2})}\Delta_{p} \bigl[\tilde{W}^{(\eta)}_{N}
\bigr](\tau_{2}) \bigr\}\,d\tau_{1}\,d\tau_{2}
\\
&&\qquad\qquad =\varepsilon^{-2}\int_{\delta_{0}d_{N}}^{\infty}\int
_{\delta
_{0}d_{N}}^{\infty}\int_{-\infty}^{\infty}
\int_{-\infty}^{\infty
}h(\tau_{1})\overline{h(
\tau_{2})}\prod_{j=1}^{2}
\bigl(q(x_{1},\tau_{j})-q(x_{2},
\tau_{j}) \bigr)
\nonumber
\\[-8pt]
\\[-8pt]
\nonumber
&&\qquad\qquad\quad {}\times K_{N}^{2}(x_{1},x_{2})
\,dx_{1}\,dx_{2}\,d\tau_{1}\,d\tau_{2},
\end{eqnarray}
where we computed the expectation using the identity (\ref{orig}) and
%
%
\begin{equation}
q(x,\tau) = -\log\biggl|x-x_{0}-\frac{\tau+i\eta}{d_{N}} \biggr|+\log
\biggl|x-x_{0}-\frac{\tau+p+i\eta}{d_{N}} \biggr|.
\end{equation}
%
Now, since $h$ is a Schwartz test function, we know that for any
$\gamma>0$ and $u>0$, we have $|h(ud_{N})| \leq(d_{N}u)^{-\gamma}$
for $N$ large enough. Then using the inequalities $|q(x,\tau)|\leq
C_{p,\eta}$ for some finite constant depending only on $p$ and $\eta
$, $K_{N}^{2}(x_{1},x_{2}) \leq N^{2}\rho_{N}(x_{1})\rho_{N}(x_{2})$
and substituting $\tau_{j}=ud_{N}$ we obtain
%
%
\begin{equation}
\label{tailineq}
\mathbb{P}\bigl(|\mathcal{I}_{3}|>\varepsilon\bigr) \leq4
\varepsilon^{-2} C_{p,\eta}^{2}N^{2}\,d_{N}^{-2\gamma+2}
\biggl(\int_{\delta
_{0}}^{\infty}u^{-\gamma}\,du
\biggr)^{2}.
\end{equation}
Then provided $d_{N}$ takes the form $d_{N} = N^{\alpha}$ with $\alpha
\in(0,1)$ we can always choose $\gamma>0$ large enough such that the
RHS of (\ref{tailineq}) tends to $0$ as $N \to\infty$.
\end{pf}
We can now translate the result (\ref{convhschwartz}) into a statement
about the Fourier coefficients $b_{N}(s)$, allowing us to prove
Theorem~\ref{thfourierconv}. For the convenience of the reader, we
repeat the
statement of the latter result here.
%

\begin{theorem}
\label{propconvcomp}
Let $\xi_{1},\ldots,\xi_{m}$ be smooth
functions compactly supported on $\mathbb{R}_{+}$. Then the vector
$(c_{N}(\xi_{1}),\ldots,c_{N}(\xi_{m}))$ converges in distribution
to a centered complex Gaussian vector $Z$ with relation matrix
$C=\mathbb{E}\{ZZ^{\mathrm{T}}\}=0$ and covariance matrix $\Gamma=
\mathbb{E}\{ZZ^{\dagger}\}$ given by
%
%
\begin{equation}
\Gamma_{j,k} = \int_{0}^{\infty}
\xi_{j}(s)\overline{\xi_{k}(s)}\, ds, \qquad j,k=1,
\ldots,m.
\end{equation}
\end{theorem}

\begin{pf}
Define functions $h_{k}$ in terms of their Fourier transform as
%
%
\begin{equation}
\label{hdef2} \int_{-\infty}^{\infty}h_{k}(
\tau)e^{-i\tau s}\,d\tau= \frac
{\sqrt{s}}{1-e^{-ips}}e^{\eta s}
\xi_{k}(s), \qquad k=1,\ldots,m.
\end{equation}
Then for sufficiently small $p$, the RHS of (\ref{hdef2}) is smooth
and compactly supported. Therefore, its Fourier transform $h_{k}$ is a
Schwartz function, that is, $h_{k} \in S(\mathbb{R})$. Next, note that
with $c_{N}(\xi)$ as in (\ref{cn2}), we have the identity
%
%
\begin{equation}
c_{N}(\xi_{k})-\mathbb{E} \bigl(c_{N}(
\xi_{k}) \bigr) = 2\int_{-\infty
}^{\infty}h_{k}(
\tau)\Delta_{p} \bigl[\tilde{W}^{(\eta)}_{N} \bigr](
\tau)\, d\tau
\end{equation}
which holds almost surely and follows after inserting the
representation (\ref{detforident2}) and interchanging the order of
integration, justified by the rapid decay of $\xi_{k}$ and $h_{k}$.
Now we apply Proposition~\ref{proptailsh} with $h(\tau)=\sum
_{k=1}^{m}\alpha_{k}h_{k}(\tau)$ where $\alpha_{k} \in\mathbb{C}$.
Since $\mathbb{E}(c_{N}(\xi_{k}))=O(d_{N}^{-1})$, we get the
convergence in distribution
%
%
\begin{equation}\label{limitcm}
\qquad\sum_{k=1}^{m}\alpha_{k}c_{N}(
\xi_{k}) \stackrel{d} {\Longrightarrow} 2\sum
_{k=1}^{m}\alpha_{k}\int
_{-\infty}^{\infty}h_{k}(\tau)
\Delta_{p} \bigl[B^{(\eta)}_{0} \bigr](\tau)\,d\tau,
\qquad N \to\infty.
\end{equation}
By the Cram\'er--Wold device, this implies the convergence in distribution
%
%
\begin{equation}
\bigl(c_{N}(\xi_{1}),\ldots,c_{N}(
\xi_{k}) \bigr) \stackrel{d} {\Longrightarrow} \bigl(Z(h_{1}),
\ldots,Z(h_{m}) \bigr),
\end{equation}
where
%
%
\begin{equation}
Z(h_{k}) = 2\int_{-\infty}^{\infty}h_{k}(
\tau)\Delta_{p} \bigl[B^{(\eta
)}_{0} \bigr](\tau)
\,d \tau.
\end{equation}
Since $\Delta_{p}[B^{(\eta)}_{0}](\tau)$ is a Gaussian process, one
easily sees that $(Z(h_{1}),\ldots,Z(h_{m}))$ is a mean zero complex
Gaussian vector. Then by a simple computation using the integral
representation (\ref{b0inc}) and basic properties of the white noise
measure $B_{c}(ds)$, we find the covariance structure
%
%
\begin{equation}
\Gamma_{j,k}=\mathbb{E} \bigl\{Z(h_{j})
\overline{Z(h_{k})} \bigr\} = \int_{0}^{\infty}
\xi_{j}(s)\overline{\xi_{k}(s)}\,ds,
\end{equation}
and $C_{j,k}=\mathbb{E}\{Z(h_{j})Z(h_{k})\}=0$ for all $j,k=1,\ldots,m$.
\end{pf}

\section{Macroscopic regime}
\label{seweakconv}
The main goal of this section is to prove Theorem~\ref{thglobal}.
Namely, we will show that the process $\tilde{D}_{N}(x)$ (\ref
{logdet}) converges in probability law as $N \to\infty$ to the
generalized Gaussian process $F(x)$ given by (\ref{1fch}). The
convergence is interpreted in the Sobolev space $V^{(-a)}$, that is, the assertion of Theorem~\ref{thglobal} is that for any
bounded continuous functional $q$ on $V^{(-a)}$, we have
%
%
\begin{equation}
\label{eqweakconveq} \lim_{N\to\infty} \mathbb{E} \bigl\{ q(
\tilde{D}_{N} ) \bigr\}= \mathbb{E} \bigl\{ q(F) \bigr\}.
\end{equation}
Our proof is an adaptation for the GUE matrices ${\mathcal H}$ of the proof
of a similar result for the CUE matrices given in \cite{HKOC01}.
First, we will prove that the finite-dimensional distributions of
$\tilde D_N(x)$ converge to those of $F(x)$ and then establish that the
sequence $\tilde D_N(x)$ is tight in $V^{(-a)}$. This will imply the
convergence in probability law in $V^{(-a)}$ as in (\ref
{eqweakconveq}). As explained in Section~\ref{sec2.2}, for the GUE
matrices there are additional analytical complications compared with
the case of CUE matrices.

We start with a deterministic result, writing down the
Chebyshev--Fourier series for $\tilde D_N(x)$.


\begin{lemma}
\label{lechebycoeffs} Let $\mathcal H$ be a Hermitian matrix of size
$N\times N$ with eigenvalues $x_1, \ldots, x_N$. Then
\[
-\log\bigl|\det(\mathcal{H}-xI)\bigr| = N\log2+ \sum_{k=0}^{\infty}
c_k (D_N) T_k(x),
\]
where the convergence is pointwise for any $x\in[-1,1]\setminus\{
x_1, \ldots, x_N \}$ and the Chebyshev--Fourier coefficients
$c_{k}({D}_{N})$ are given for any $k > 0$ by the formula
%
%
\begin{equation}
\label{insertcoeffs}
c_{k}(D_{N}) = \sum
_{j=1}^{N}\frac{2}{k}T_{k}(x_j)+
\sum_{j=1}^{N} r^{+}_{k}(x_{j})+
\sum_{j=1}^{N}r^{-}_{k}(x_{j})
\end{equation}
and
%
%
\begin{equation}
\label{insertcoeffs2} c_{0}(D_{N}) = -\sum
_{j=1}^{N}r^{+}_{0}(x_{j})
-\sum_{j=1}^{N}r^{-}_{0}(x_{j}),
\end{equation}
where for $k > 0$
%
%
\begin{equation}
\label{testfns} 
r^{\pm}_{k}(x) =
\bigl[(2/k) (-T_{k}(x) + \bigl(x \mp\sqrt{x^{2}-1}
\bigr)^{k} \bigr]\chi_{(\pm1, \pm\infty)}(x)
\end{equation}
and
%
%
\begin{equation}
r^{\pm}_{0}(x) = \log\bigl|x \mp
\sqrt{x^{2}-1}\bigr| \chi_{(\pm1, \pm
\infty)}(x).
\end{equation}
In the above formulae, $\chi_{J}(x)$ is the indicator function on the
set $J$. 
\end{lemma}

\begin{pf}
This follows immediately from Lemma~3.1 in \cite{GP13}.
\end{pf}

It follows from this lemma that for our random matrices $\mathcal H$, with
probability one,
\[
\tilde{D}_{N}(x)= \sum_{k=0}^{\infty}
c_k (\tilde D_N) T_k(x) \qquad
\mbox{where }  c_k (\tilde D_N)
=c_{k}({D}_{N})- \mathbb{E} \bigl\{ c_{k}({D}_{N})
\bigr\}.
\]

\subsection{Convergence of finite-dimensional distributions}
The main goal of this subsection is to establish the following.
%

\begin{proposition}
\label{propfindim}
Fix $M \in\mathbb{N}$ and let $X_{1},\ldots,X_{M}$ be independent
Gaussian random variables with mean zero and variance one. Then for any
$(t_{k})_{k=1}^{M} \in\mathbb{R}^{M}$ we have the convergence in distribution
%
%
\begin{equation}\label{chebydist}
\sum_{k=0}^{M}c_{k}(
\tilde{D}_{N})t_{k} \stackrel{d} {\Longrightarrow} \sum
_{k=1}^{M}\frac{X_{k}}{\sqrt{k}}t_{k},
\qquad N \to\infty.
\end{equation}
\end{proposition}

\begin{pf}
We begin by inserting equation (\ref{insertcoeffs}) into the LHS of
(\ref{chebydist}). Then from \cite{Joh98} or \cite{PS11}, we know
that the sum
%
%
\begin{equation}
\sum_{k=1}^{M}t_{k} \Biggl(
\sum_{j=1}^{N}\frac
{2}{k}T_{k}(x_{j})-
\mathbb{E} \Biggl\{\sum_{j=1}^{N}
\frac
{2}{k}T_{k}(x_{j}) \Biggr\} \Biggr)
\end{equation}
converges in distribution to the RHS of (\ref{chebydist}) as $N \to
\infty$. The main technical part of our proof of (\ref{chebydist})
consists in showing that the other terms appearing in (\ref
{insertcoeffs}) and~(\ref{insertcoeffs2}) do not contribute in the
limit $N \to\infty$. All such terms that appear are of the form
%
%
\begin{equation}
\label{linstatsing} A^{\pm}_{k,N} = \sum
_{j=1}^{N}r^{\pm}_{k}(x_{j})
\end{equation}
and by definition of the test function $r^{\pm}_{k}(x)$, they are
nonzero only when an eigenvalue $x_{j}$ lies outside the bulk of the
limiting spectrum $[-1,1]$. Intuitively, this is a rare event and we
show below that in fact $\mathbb{E}|A^{\pm}_{k,N}| \to0$ as $N \to
\infty$. We note in passing that the regularity of the test functions
$r^{\pm}_{k}(x)$ lies outside the best known $C^{1/2+\varepsilon}$
threshold in \cite{SW13}, due to the singularities at the spectral edges.

Let us focus our attention on the case $\mathbb{E}\{|A^{+}_{k,N}|\}$,
since the estimation of $\mathbb{E}\{|A^{-}_{k,N}|\}$ follows exactly
the same pattern. First, one sees from the explicit formula (\ref
{testfns}) and the elementary inequality $(x-\sqrt{x^{2}-1})^{k} \leq
T_{k}(x) \leq(x+\sqrt{x^{2}-1})^{k}$, $x\geq1$ that $-r^{+}_{k}(x)$
is nonnegative for all $x \in\mathbb{R}$. Therefore, $\mathbb{E}\{
|A^{+}_{k,N}|\} = -\mathbb{E}\{A^{+}_{k,N}\}$. 

In terms of the normalized eigenvalue density, we have 
%
%
\begin{equation}
\mathbb{E} \bigl\{A^{+}_{k,N} \bigr\} = N\int
_{1}^{\infty}r^{+}_{k}(x)
\rho_{N}(x)\,dx.
\end{equation}
%
To proceed, we split the integral as
%
%
\begin{equation}
\label{edgeandouter}
\quad\mathbb{E} \bigl\{A^{+}_{k,N} \bigr\} = N
\int_{1}^{1+\delta_{N}}r^{+}_{k}(x)
\rho_{N}(x)\,dx+N\int_{1+\delta_{N}}^{\infty}r^{+}_{k}(x)
\rho_{N}(x)\,dx,
\end{equation}
where we choose $\delta_{N}=N^{-7/12}$. The first integral in (\ref
{edgeandouter}) is over a shrinking neighbourhood of the spectral edge
$x=1$. An estimate that holds uniformly in this region can be given in
terms of the Airy function $\mathrm{Ai}(x)$ and its derivatives. In
particular, equation (4.4) of \cite{EM03} (see also the Proof
of Lemma~2.2 in \cite{G05}) shows that as $N \to\infty$
%
%
\begin{eqnarray}
%
N\rho_{N}(x) &=& \biggl(\frac{\Phi'(x)}{4\Phi(x)}-
\frac{\gamma
'(x)}{\gamma(x)} \biggr) \bigl[2\mathrm{Ai} \bigl(N^{2/3}\Phi(x)
\bigr) \mathrm{Ai}' \bigl(N^{2/3}\Phi(x) \bigr) \bigr]
\nonumber
\\
\label{airyasy}
&&{}+N^{2/3}\Phi'(x) \bigl[ \bigl(\mathrm{Ai}'
\bigl(N^{2/3}\Phi(x) \bigr) \bigr)^{2}-N^{2/3}\Phi
(x) \bigl(\mathrm{Ai} \bigl(N^{2/3}\Phi(x) \bigr) \bigr)^{2}
\bigr]\\
\nonumber
&&{}+O \biggl(\frac{1}{N(\sqrt
{x-1})} \biggr),  
\end{eqnarray}
where
%
%
\begin{equation}
\gamma(x) = \biggl(\frac{x-1}{x+1} \biggr)^{1/4}
\end{equation}
and
%
%
\begin{equation}
\Phi(x) = %
\cases{\displaystyle - \biggl(3\int_{x}^{1}
\sqrt{1-y^{2}}\,dy \biggr)^{2/3}, & \quad$|x|\leq1$,
\vspace*{3pt}
\cr
\displaystyle\biggl(3\int_{1}^{x}
\sqrt{y^{2}-1}\,dy \biggr)^{2/3}, & \quad$|x|>1$.}
\end{equation}
Since $\Phi(x)\geq0$ for $x\geq1$, the functions\vspace*{1pt} $\mathrm{Ai}(N^{2/3}\Phi(x))$ and $\mathrm{Ai}'(N^{2/3}\Phi(x))$ are
uniformly bounded on $[1,\infty)$. Furthermore,\vspace*{1pt} $ (\frac{\Phi
'(x)}{4\Phi(x)}-\frac{\gamma'(x)}{\gamma(x)} )$ and $\Phi
'(x)$ are bounded near $x=1$. Inserting (\ref{airyasy}) into the first
integral in (\ref{edgeandouter}), we obtain the bound
%
%
\begin{equation}
\label{twoints} N\int_{1}^{1+\delta_{N}}r^{+}_{k}(x)
\rho_{N}(x)\,dx = c_{1}N^{2/3}\int
_{1}^{1+\delta_{N}}r^{+}_{k}(x)
\,dx+O \biggl(\frac
{1}{N} \biggr),
\end{equation}
where $c_{1}$ is an $N$-independent constant. In (\ref{twoints}), we
used that $r^{+}_{k}(x)(x-1)^{-1/2}$ is bounded near $x=1$ to estimate
the contribution of the error term in (\ref{airyasy}). A~simple
computation shows that $\int_{1}^{1+\delta_{N}}r^{+}_{k}(x)\,
dx=O(\delta_{N}^{3/2})$ as $N \to\infty$ for $k \geq0$. Inserting
the latter into (\ref{twoints}) yields the bound
%
%
\begin{equation}
N\int_{1}^{1+\delta_{N}}r^{+}_{k}(x)
\rho_{N}(x)\,dx = O \bigl(N^{2/3}\delta_{N}^{3/2}
\bigr) = O \bigl(N^{-5/24} \bigr).
\end{equation}
Now consider the second integral in (\ref{edgeandouter}). We will
prove below that it is exponentially small as $N \to\infty$. Using
the fact that (for $k\geq1$) $-r^{+}_{k}(x) \leq T_{k}(x)$ and
applying Lemma~\ref{leexpdecaylem}, we obtain
%
%
\begin{eqnarray}\label{cint1}
&&-N\int_{1+\delta_{N}}^{\infty}r^{+}_{k}(x)
\rho_{N}(x)\,dx
\\
&&\qquad \leq N\delta_{N}\int_{1}^{\infty}T_{k}(1+u
\delta_{N})\rho_{N}(1+u\delta_{N})\,du
\\
\label{cint}
&&\qquad \leq B^{-1}\int_{1}^{\infty}u^{-1}T_{k}(1+u
\delta_{N})e^{-buN^{1/8}}\,du,
\end{eqnarray}
where $B,b>0$ are absolute constants. Then, for example,
expanding $T_{k}(1+u\delta_{N})$ in powers of $(u\delta_{N})$ and
integrating (\ref{cint}) term\vspace*{1pt} by term, we can apply the standard
Laplace method and find that (\ref{cint}) is $O(e^{-cN^{1/8}})$ for
some $c>0$. If $k=0$ in the integral (\ref{cint1}), one can use the
inequality $|r^{+}_{0}(1+x)| \leq\sqrt{2x}$, $x>0$ and then apply the
Laplace method as before yielding a similar error bound. This completes
the proof of the proposition.
\end{pf}

\subsection{Tightness}
The final\vspace*{1pt} ingredient required for proving the weak convergence in
(\ref{eqweakconveq}) is to show that the sequence $\tilde{D}_{N}$ is
tight in $V^{(-a)}$. In direct analogy to the proof given in
Theorem~2.5 of \cite{HKOC01} for the Circular Unitary Ensemble, we will
exploit the convenient fact that for $-\infty< a < b < \infty$, the
closed unit ball in $V^{(b)}$ is compact in $V^{(a)}$. Then by
Chebyshev's inequality, tightness follows if we can bound the variance
%
%
\begin{equation}
\mathbb{E}\Vert \tilde{D}_{N}\Vert^{2}_{(-b)} = \sum
_{k=0}^{\infty
}\mathbb{E} \bigl
\{c_{k}(\tilde{D}_{N})^{2} \bigr\}
\bigl(1+k^{2} \bigr)^{-b}
\end{equation}
uniformly in $N$. Such a uniform bound will follow for any $b > 1/2$
provided we show that $\mathbb{E}\{c_{k}(\tilde{D}_{N})^{2}\} \leq
C$ for some constant $C$ independent of $k$ and $N$. We begin by
writing the Chebyshev--Fourier coefficient as
%
%
\begin{equation}
c_{k}(\tilde{D}_{N}) = \sum_{j=1}^{N}h_{k}(x_{j})-
\mathbb{E} \Biggl\{ \sum_{j=1}^{N}h_{k}(x_{j})
\Biggr\},
\end{equation}
where
%
%
\begin{eqnarray}
%
h_{k}(x) &=& (2/k)T_{k}(x)
\chi_{[-1,1]}(x) - (2/k) \bigl(x-\sqrt{x^{2}-1}
\bigr)^{k}\chi_{(1,\infty)}(x)
\nonumber
\\[-8pt]
\\[-8pt]
\nonumber
&&{}-(2/k) \bigl(x+\sqrt{x^{2}-1} \bigr)^{k}
\chi_{(-1,-\infty)}(x). 
%
\end{eqnarray}
Then by formula (\ref{orig}), we have
%
%
\begin{equation}\label{varform}
\mathbb{E} \bigl\{c_{k}(\tilde{D}_{N})^{2} \bigr
\} = \frac{1}{8}\int_{\mathbb
{R}^{2}} \bigl(h_{k}(x_{1})-h_{k}(x_{2})
\bigr)^{2}K_{N}(x_{1},x_{2})^{2}
\,dx_{1}\, dx_{2},
\end{equation}
where $K_{N}(x,y)$ is the GUE kernel defined in equation (\ref{guekernel}).

First, we consider the contribution to the integral (\ref{varform})
coming from the region $[-1,1]^{2}$, namely the integral
%
%
\begin{equation}\label{varform2}
\frac{1}{2k^{2}}\int_{[-1,1]^{2}} \biggl(\frac{\Delta
T_{k}(x)}{\Delta x}
\biggr)^{2} \mathcal{F}_{N}(x_{1},x_{2})
\,dx_{1}\, dx_{2},
\end{equation}
where $\mathcal{F}_{N}(x_{1},x_{2})$ is defined by (\ref{cfn}) and,
as in Section~\ref{sefouriercoeff}, for a function $f$, we denote by $\Delta f$ the
difference $\Delta f(x) = f(x_{1})-f(x_{2})$. By the Plancherel--Rotach
asymptotics of Hermite polynomials, we have the bound (as follows from,
e.g., parts (iii) and (v) of Theorem~2.2 in \cite{DKMVZ99})
%
%
\begin{equation}
\bigl|\mathcal{F}_{N}(x_{1},x_{2})\bigr| \leq
\frac{K_{1}}{\sqrt
{1-x_{1}^{2}}\sqrt{1-x_{2}^{2}}}
\end{equation}
uniformly for $(x_{1},x_{2}) \in[-1,1]^{2}$. This implies that the
modulus of (\ref{varform2}) is bounded by
%
%
\begin{equation}\label{explicitint}
\frac{K_{1}}{2k^{2}}\int_{[-1,1]^{2}} \biggl(\frac{\Delta
T_{k}(x)}{\Delta x}
\biggr)^{2}\frac{1}{\sqrt{1-x_{1}^{2}}\sqrt
{1-x_{2}^{2}}}\,dx_{1}\,dx_{2}=K_{1}
\pi^{2}/8.
\end{equation}
The equality in (\ref{explicitint}) is a simple exercise involving
standard properties of Chebyshev polynomials and we omit the derivation.

Finally, consider the contribution to the integral (\ref{varform})
from outside the square $[-1,1]^{2}$. For simplicity, consider just the
region $1 < x_{1} < \infty$ and $-1<x_{2}<1$, all others being
analogous. Since $h_{k}(x)$ is uniformly bounded in $k$ and $x$ on the
whole real line, we have
%
%
\begin{eqnarray}
&&\int_{-1}^{1}\int_{1}^{\infty
}
\bigl(h_{k}(x_{1})-h_{k}(x_{2})
\bigr)^{2}K_{N}(x_{1},x_{2})^{2}
\,dx_{1}\,dx_{2}
\\
&&\qquad \leq\int_{-\infty}^{\infty}\int_{1}^{\infty
}K_{N}(x_{1},x_{2})^{2}
\,dx_{1}\,dx_{2}
\\
\label{lastint}
&&\qquad =\int_{1}^{\infty}N\rho_{N}(x_{1})
\,dx_{1}= \int_{1}^{1+\delta
}N
\rho_{N}(x_{1})\,dx_{1}+O \bigl(Ne^{-c_{\delta}N}
\bigr),
\end{eqnarray}
where $\delta>0$ is a constant and $c_{\delta}>0$. The last equality
in (\ref{lastint}) follows from Theorem~5.2.3(iii) in \cite{PS11}.
Now we can insert the formula (\ref{airyasy}) which holds uniformly on
$[1,1+\delta]$. The first term in (\ref{airyasy}) is bounded in $N$
and $x_{1}$ and so its integral over $[1,1+\delta]$ is bounded in $N$.
The third term gives an error of order $1/N$. The contribution from the
middle term can be explicitly integrated using the substitution
$u=N^{2/3}\Phi(x_{2})$:
%
%
\begin{eqnarray}
&&\qquad\int_{1}^{1+\delta}N^{2/3}
\Phi'(x_{2}) \bigl(\mathrm{Ai}'^{2}
\bigl(N^{2/3}\Phi(x_{2}) \bigr)-N^{2/3}
\Phi(x_{2})\mathrm{Ai}^{2} \bigl(N^{2/3}
\Phi(x_{2}) \bigr) \bigr)\,dx_{2}
\\
&&\qquad\qquad=\int_{0}^{N^{2/3}\Phi(1+\delta)} \bigl[\mathrm{Ai}'^{2}(u)-u
\mathrm{Ai}^{2}(u) \bigr]\,du
\\
&&\qquad\qquad=- \biggl[\frac{2}{3} \bigl(u^{2}\mathrm{Ai}^{2}(u)-u
\mathrm{Ai}'^{2}(u) \bigr)-\frac{1}{3}
\mathrm{Ai}(u)\mathrm{Ai}'(u) \biggr]_{0}^{N^{2/3}\Phi(1+\delta)}
\\
&&\qquad\qquad =\mathrm{Ai}(0)\mathrm{Ai}'(0)/3+O \bigl(e^{-d_{\delta}N} \bigr),
\end{eqnarray}
where $d_{\delta}>0$. A completely\vspace*{1pt} analogous argument proves that the
integral over the region $\{1<x_{1}<\infty, 1<x_{2}<\infty\}$ is also
uniformly bounded in $k$ and $N$, in addition to the remaining $6$
regions that make up $B^{c}$. This completes the proof that $\tilde
{D}_{N}$ is tight in $V^{(-a)}$ for any $a > 1/2$, and hence completes
the proof of Theorem~\ref{thglobal}.

\begin{appendix}\label{app}
\section{Proof of Proposition \texorpdfstring{\protect\ref{11111}}{3.2}}
\label{apriemann}
The purpose of this Appendix is to give the technical details required
to show that the matrix $P_{\infty}(z)$ in Section~\ref{seouter}
gives a good approximation to the matrix $S(z)$ in Section~\ref
{sesmatrix} for large $N$, as described by Proposition~\ref
{11111}. Although we can mostly follow the now standard techniques
described in \cite{DKMVZ99}, we must take special care with the
estimates because the system of contours in Figure~\ref{figcontour}
can come arbitrarily close to the real axis as $N \to\infty$.
%

\begin{remark}
In this Appendix, there are many estimates holding uniformly in the
parameters\vspace*{1pt} $\{\tau_{k}\}_{k=1}^{m-1}$, $\{\alpha_{k}\}_{k=1}^{m-1}$
and $x_{0}$ that appear in the partition function~(\ref{multint}). We
will use the big-oh notation $\mathcal{O}$ (distinguished from the
usual $O$) for an error term that defines an analytic function of the
parameters $\{\alpha_{k}\}_{k=1}^{m-1}$ on $\Omega$ [cf. (\ref
{omset})] satisfying uniformity in the following parameters:
\begin{itemize}
\item $\tau_{k}$ varying in a compact subset of $\mathbb{R}$ for
$k=1,\ldots,m-1$,
\item $\alpha_{k}$ varying in a compact subset of $\Omega$ for
$k=1,\ldots,m-1$,
\item $x_{0}$ varying in a compact subset of $(-1+\delta,1-\delta)$.
\end{itemize}
\end{remark}

\subsection*{Construction of the parametrices at \texorpdfstring{$z=\pm1$}{z=+-1}}
The parametrices at $z=\pm1$ consist of a matrix valued function
$P_{\pm1}(z)$ defined in the discs $B_{\pm1}(\delta)$ (cf.
Figure~\ref{figcontour}) satisfying the following properties:
\begin{longlist}[3.]
\item[1.] $P_{\pm1}(z)$ is analytic in $B_{\pm1}(\delta)\setminus\Sigma$.
%
\item[2.] $P_{\pm1}(z)$ satisfies the same jump conditions as $S(z)$ on
$\Sigma\cap B_{\pm\delta}$.
\item[3.] The following matching condition is satisfied on the boundary
$\partial B_{\pm1}(\delta)$:
%
%
\begin{equation}
\label{matching}
P_{\pm1}(z)P_{\infty}(z)^{-1} = I+O
\bigl(N^{-1} \bigr), \qquad z \in\partial B_{\pm1}(\delta),
\end{equation}
as $N \to\infty$.
\end{longlist}

The functions $P_{1}(z)$ and $P_{-1}(z)$ can be obtained in precisely
the same way as in~\cite{K07}, which was itself based on the
construction in \cite{DKMVZ99} corresponding to weights $\omega(z)
\equiv1$. In our situation, the only difference is that our weight
$\omega(z)$ and the Szeg\"o function $\mathcal{D}(z)$ are
$N$-dependent, so that one has to be careful with the matching
condition (\ref{matching}). From equation (76) in \cite{K07}, we have
%
%
\begin{eqnarray}
&&\qquad P_{\pm1}(z)P_{\infty}(z)^{-1}
\nonumber
\\[-8pt]
\label{parametrix}
\\[-8pt]
\nonumber
&&\qquad\qquad =
P_{\infty}(z)\omega(z)^{\sigma
_{3}/2}\tilde{P}_{\infty}(z)^{-1}
\tilde{P}_{\pm1}(z)\tilde{P}_{\infty}(z)^{-1}
\tilde{P}_{\infty}(z)\omega(z)^{-\sigma
_{3}/2}P_{\infty}(z)^{-1},
\end{eqnarray}
where\vspace*{1pt} $\tilde{P}_{\pm1}(z)$ and $\tilde{P}_{\infty}(z)$ are the
quantities $P_{\pm1}(z)$ and $P_{\infty}(z)$ with $\omega(z)\equiv
1$. For our purposes, we will not need the explicit expression for
$\tilde{P}_{\pm1}(z)$, which can be found in, for example,
\cite{DKMVZ99} or \cite{K07}. Our main goal here is to check that the
matching condition (\ref{matching}) is still satisfied.

%
\begin{lemma}
\label{lediscestlem}
Let $P_{\pm1}(z)$ denote the parametrix defined in (\ref
{parametrix}). Then we have as $N \to\infty$
%
%
\begin{equation}
P_{\pm1}(z)P_{\infty}(z)^{-1}=I+\frac{\tilde{\Delta}_{1}^{(\pm
1)}(z)}{N}+
\mathcal{O} \biggl(\frac{1}{Nd_{N}} \biggr), \qquad z \in\partial
B_{\pm1}(\delta),
\end{equation}
where the estimate is uniform for $z \in\partial B_{\pm1}(\delta)$.
The first correction term $\tilde{\Delta}_{1}^{(\pm1)}(z)$ depends
only on $z$ and is analytic except for a second-order pole at $z = \pm1$.
\end{lemma}

\begin{pf}
Proposition~7.7 of \cite{DKMVZ99} implies that there is a uniform
asymptotic expansion
%
%
\begin{equation} \label{deiftexp}
\tilde{P}_{\pm1}(z)\tilde{P}_{\infty}(z)^{-1} \sim I +
\sum_{k=1}^{\infty}\frac{\tilde{\Delta}^{(\pm1)}_{k}(z)}{N^{k}},
\qquad z
\in\partial B_{\pm1}(z),
\end{equation}
where $\tilde{\Delta}^{(\pm1)}_{k}(z)$ are independent of $N$ [and
independent of $\omega(z)$], and have meromorphic continuations inside
the disc $\partial B_{\pm1}(\delta)$ with a pole of order $(3k+1)/2$
at $z=\pm1$. Inserting (\ref{deiftexp}) back into (\ref
{parametrix}), we find that
%
%
\begin{equation}
\label{asympt} \quad\hspace*{5pt} P_{\pm1}(z)P_{\infty}(z)^{-1}-I \sim
\sum_{k=1}^{\infty}\frac
{Q(z)\tilde{\Delta}^{(\pm1)}_{k}(z)Q(z)^{-1}}{N^{k}}, \qquad z
\in\partial B_{\pm1}(\delta),
\end{equation}
where $Q(z) = P_{\infty}(z)\omega(z)^{\sigma_{3}/2}\tilde
{P}_{\infty}(z)^{-1}$. To prove the lemma, it is sufficient to show that
%
%
\begin{equation} \label{qasympt}
Q(z) = I+\mathcal{O} \bigl(d_{N}^{-1} \bigr), \qquad z \in
B_{\pm1}(\delta).
\end{equation}
First, note that
%
%
\begin{equation}\label{para1est}
\omega(z) = 1+\mathcal{O} \bigl(d_{N}^{-1} \bigr), \qquad z
\in\partial B_{\pm
1}(\delta)\cup[-1,1]
\end{equation}
as follows immediately from the representation (\ref{omegaweight}).
Then the proof is complete if we can check that
%
%
\begin{equation}\label{szest}
\qquad\frac{\sqrt{z-1}\sqrt{z+1}}{2\pi}\int_{-1}^{1}
\frac{\log\omega
(x)}{\sqrt{1-x^{2}}(z-x)}\,dx = \mathcal{O} \bigl(d_{N}^{-1} \bigr),
\qquad z \in\partial B_{\pm1}(\delta)
\end{equation}
because this would imply the corresponding estimate for the Szeg\"o\vspace*{2pt}
function $\mathcal{D}(z) = 1+\mathcal{O}(d_{N}^{-1})$ [cf. (\ref
{sz})] so that $P_{\infty}(z)=\tilde{P}_{\infty}(z)+\mathcal
{O}(d_{N}^{-1})$. We will prove (\ref{szest}) below only for $z \in
\partial B_{1}(\delta)$, the case $z \in\partial B_{-1}(\delta)$
being identical. If $(z-x)^{-1}$ is bounded, the result follows
immediately from (\ref{para1est}), therefore, we consider only the
contribution to the integral (\ref{szest}) from a small neighbourhood
$[1-\delta-\varepsilon_{0},1-\delta+\varepsilon_{0}]$ and the points $z
\in\partial B_{1}(\delta)$ such that $0<|z-(1-\delta)|<\varepsilon
_{0}/2$. First, consider $\operatorname{Im}(z)>0$ and let $\mathcal
{C}$ denote the
clockwise oriented semi-circle in the upper-half plane connecting the
points $1-\delta-\varepsilon_{0}$ and $1-\delta+\varepsilon_{0}$. Then by
the residue theorem and analyticity of $\omega(x)$, (\ref{szest}) is
equal to
%
%
\begin{equation}\label{rescalc}\quad
i\sqrt{z+1}\sqrt{z-1}\frac{\log\omega(z)}{\sqrt{1-z^{2}}}+\frac
{\sqrt{z-1}\sqrt{z+1}}{2\pi}\int
_{\mathcal{C}}\frac{\log\omega
(x)}{\sqrt{1-x^{2}}(x-z)}\,dx,
\end{equation}
where we take the principal branch of the square root. Now both terms
in (\ref{rescalc}) are clearly $O(d_{N}^{-1})$, as follows\vspace*{1pt} from (\ref
{para1est}) and the fact that $(x-z)^{-1}$ is uniformly bounded in
(\ref{rescalc}). A similar calculation applies when $\operatorname
{Im}(z)<0$. This
completes the proof of the lemma.
\end{pf}
\subsection*{Final transformation}
\label{secontours}
We will now define the final transformation of the Riemann--Hilbert
problem, $S \to R$. As usual, we set
%
%
\begin{eqnarray}
\label{rsp} R(z) &= &
\cases{ S(z)P_{\infty}(z)^{-1},
& \quad $z \in U_{\infty}\setminus\Sigma$,\vspace*{3pt}
\cr
S(z)P_{\pm1}(z)^{-1}, & \quad$z \in B_{\pm1}(\delta)
\setminus\Sigma$.}%
\end{eqnarray}
%
%
\begin{figure}

\includegraphics{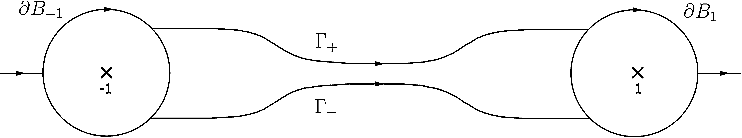}

\caption{The contour $\Sigma_{R}$ for the $R(z)$ Riemann--Hilbert
problem. The parts of the lenses $\Gamma=\Sigma\setminus\partial
B_{\pm1}(\delta)$ near $x_{0}$ are of distance $O(d_{N}^{-1})$ from
the real line. The circles $\partial B_{\pm1}(\delta)$ are of radius
$\delta$.}
\label{figrcontour}
\end{figure}
%

\noindent From the Riemann--Hilbert problem for $S(z)$, it is easily shown that
$R(z)$ has jumps only on $\partial B_{\pm1}(\delta)$, $\mathbb
{R}\setminus[-1-\delta,1+\delta]$ and the parts of $\Sigma_{\pm}$
outside of $B_{1}(\delta)\cup B_{-1}(\delta)$ (denoted here by
$\Gamma_{\pm}$). In what follows, we will denote the disjoint union
of these contours as $\Sigma_{R}$, which we plot in Figure~\ref
{figrcontour}. The function $R(z)$ satisfies the following:
\begin{longlist}[3.]
\item[1.] $R(z)$ is analytic in $\mathbb{C}\setminus\Sigma_{R}$.
\item[2.] $R(z)$ satisfies the jump condition $R_{+}(s)=R_{-}(s)J(s)$ where
%
%
\begin{eqnarray}
J(s) &=& P_{\infty}(s) %
\pmatrix{ 1 & \omega(s)e^{N(g_{+}(s)+g_{-}(s)-2s^{2}-l)}
\vspace*{3pt}
\cr
0 & 1 }
P_{\infty}(s)^{-1},
\nonumber
\\[-8pt]
\\[-8pt]
\eqntext{\displaystyle  s \in
\mathbb{R}\setminus[-1-\delta,1+\delta],}
\\
J(s) &=& P_{\infty}(s) %
\pmatrix{ 1 & 0\vspace*{2pt}
\cr
\omega(s)^{-1}e^{\mp Nh(s)} & 1 } P_{\infty}(s)^{-1},
\qquad s \in\Gamma_{\pm},
\\
J(s) &=& P_{\pm1}(s)P_{\infty}(s)^{-1},
\qquad s \in\partial B_{\pm1}.
\end{eqnarray}
\item[3.] $R(z) = I+O(z^{-1})$ as $z \to\infty$.
\end{longlist}

\subsection*{Estimating the jump matrix \texorpdfstring{$\Delta(s)$}{Delta(s)}}
\label{sejumpest}
Before we estimate the jump matrix, we need to understand the behaviour
of $P_{\infty}(z)$ [cf. (\ref{pinf})] on the contours $\Gamma_{\pm}$.
%

\begin{lemma}
\label{leszebound}
The Szeg\"o function $\mathcal{D}(s)$ in (\ref{sz}) and its inverse
$\mathcal{D}(s)^{-1}$ are uniformly bounded on the contours $\Gamma
_{\pm}$. In fact, we have
%
%
\begin{equation}
\log\mathcal{D}(s) = \mathcal{O}(1), \qquad N \to\infty,
\end{equation}
uniformly for $s \in\Gamma_{\pm}$.
\end{lemma}

\begin{pf}
It suffices to prove that
%
%
\begin{equation}
\label{intbound}
\int_{-1}^{1}\frac{\log\omega(x)}{(s-x)\sqrt{1-x^{2}}}
\,dx = \mathcal{O}(1).
\end{equation}
We remind the reader that the weight $\omega(x)$ can be written
%
%
\begin{equation}
\omega(x) = \prod_{k=1}^{m-1} \biggl[
\frac{(x-x_{0}-\tau
_{k}/d_{N})^{2}+(\eta/d_{N})^{2}}{(x-x_{0})^{2}+(\eta
/d_{N})^{2}} \biggr]^{\alpha_{k}/2},
\end{equation}
as follows from the constraints on $\alpha_{k}$'s in (\ref{cyclic}).
We have the elementary inequality
%
%
\begin{equation}
\bigl|\log\bigl(\omega(x) \bigr)\bigr| \leq\frac{1}{2}\sum
_{k=1}^{m-1}|\alpha_{k}| \bigl|\log
\bigl(1+g_{\tau,\eta,N}(x,x_{0}) \bigr) \bigr|,
\end{equation}
where
%
%
\begin{equation}
g_{\tau,\eta,N}(x,x_{0}) = \frac{(\tau/d_{N})^{2}-2(x-x_{0})\tau
/d_{N}}{(x-x_{0})^{2}+(\eta/d_{N})^{2}}.
\end{equation}
Now, clearly if $x\leq x^{*} = x_{0}+\tau/(2d_{N})$, we have $g_{\tau
,\eta,N}(x,x_{0})\geq0$, so that $\log(1+g_{\tau,\eta,N}(x,x_{0}))
\leq g_{\tau,\eta,N}(x,x_{0})$. If $x>x^{*}$, we symmetrise about the
point $x^{*}$ exploiting the symmetry $|\log(1+g_{\tau,\eta
,N}(x^{*}-x,x_{0}))| = |\log(1+g_{\tau,\eta,N}(x^{*}+x,x_{0}))|$ to obtain
%
%
\begin{equation}
\label{2ineqs}\quad
\bigl|\log\bigl(1+g_{\tau,\eta,N}(x,x_{0}) \bigr)\bigr|
\leq \bigl|g_{\tau,\eta,N}(x,x_{0})\bigr| + \bigl|g_{\tau,\eta,N}
\bigl(2x^{*}-x,x_{0} \bigr)\bigr|.
\end{equation}
We will focus only on the region $x \in[x_{0}-\varepsilon,x_{0}+\varepsilon
]$ as this gives the dominant contribution to the integral (\ref
{intbound}). For $s \in\Gamma_{\pm}$ and $x \in[x_{0}-\varepsilon
,x_{0}+\varepsilon]$, we have $|s-x|^{-1} \leq((x-x_{0})^{2}+(\eta
/2d_{N})^{2})^{-1/2}$ and $(1-x^{2})^{-1/2} = \mathcal{O}(1)$. Then
the contribution to (\ref{intbound}) from the first term on the RHS of
(\ref{2ineqs}) is bounded by
%
%
\begin{eqnarray}
&&\int_{x_{0}-\varepsilon}^{x_{0}+\varepsilon}\frac{|g_{\tau,\eta
,N}(x,x_{0})|}{\sqrt{(x-x_{0})^{2}+(\eta/2d_{N})^{2}}}\,dx
\nonumber
\\[-8pt]
\label{extendints}
\\[-8pt]
\nonumber
&&\qquad\leq  \int
_{-1}^{1}\frac{|(\tau/d_{N})^{2}-2x\tau/d_{N}|}{(x^{2}+(\eta
/2d_{N})^{2})^{3/2}}\,dx
\\
&&\qquad= \frac{8|\tau|}{\eta} \biggl(\frac{\sqrt{\tau^{2}/\eta
^{2}+1}\sqrt{(2d_{N}/\eta)^{2}+1}-1}{\sqrt{(2d_{N}/\eta
)^{2}+1}} \biggr) = \mathcal{O}(1),
\end{eqnarray}
where we changed variables $x \to x-x_{0}$ and extended the limits of
integration back to $[-1,1]$. The resulting integral on the RHS of
(\ref{extendints}) can be evaluated exactly in, for example, Maple.

For the\vspace*{1pt} second term in (\ref{2ineqs}), we use the estimate
$((x-x_{0})^{2}+(\eta/d_{N})^{2})^{-1/2} \leq c((x_{0}-x+\tau
/d_{N})^{2}+(\eta/d_{N})^{2})^{-1/2}$ (where $c$ depends on $\eta$
and $\tau$ only) to get
%
%
\begin{eqnarray}
&& \int_{x_{0}-\varepsilon}^{x_{0}+\varepsilon}\frac{|g_{\tau,\eta
,N}(2x^{*}-x,x_{0})|}{\sqrt{(x-x_{0})^{2}+(\eta/(2d_{N}))^{2}}}\,dx
\nonumber
\\[-8pt]
\\[-8pt]
\nonumber
&&\qquad\leq   c
\int_{x_{0}-\varepsilon}^{x_{0}+\varepsilon}\frac{|(\tau
/d_{N})^{2}-2(x_{0}-x+\tau/d_{N})\tau/d_{N}|}{((x_{0}-x+\tau
/d_{N})^{2}+(\eta/(2d_{N}))^{2})^{3/2}}\,dx
\\
&&\qquad=  c\int_{-\varepsilon+\tau/d_{N}}^{\varepsilon+\tau/d_{N}}\frac{|(\tau
/d_{N})^{2}-2u\tau/d_{N}|}{(u^{2}+(\eta/(2d_{N}))^{2})^{3/2}}\,du =
\mathcal{O}(1),
\end{eqnarray}
where we used that the last integral is bounded by the RHS of (\ref
{extendints}).
\end{pf}

%
\begin{proposition}
\label{propdeltasigr}
Let $\Delta(s) = J(s)-I$ where $J(s)$ is the jump matrix for $R(z)$
defined on the contour $\Sigma_{R}$. We have the following bounds:
\begin{itemize}
\item On the discs
%
%
\begin{equation}\label{discest}
\bigl|\Delta(s)\bigr| = \mathcal{O} \bigl(N^{-1} \bigr), \qquad s \in
\partial
B_{\pm
1}(\delta).
\end{equation}
\item On the upper and lower lips
%
%
\begin{equation}\label{lipsbound}
\bigl|\Delta(s)\bigr| = \mathcal{O} \biggl(\exp\biggl(-c_{1}
\frac
{N}{d_{N}} \biggr) \biggr), \qquad s \in\Gamma_{\pm}.
\end{equation}
\item On the real line
%
%
\begin{equation}\label{realbound}
\bigl|\Delta(s)\bigr| = \mathcal{O} \bigl(\exp(-c_{2}N ) \bigr), \qquad s
\in
\mathbb{R}\setminus[-1-\delta,1+\delta].
\end{equation}
Here, $c_{1}>0$ and $c_{2}>0$ are constants depending only on $\delta$
and $\eta$.
\end{itemize}
\end{proposition}

\begin{pf}
The bound (\ref{discest}) follows immediately from Lemma~\ref
{lediscestlem}, while (\ref{realbound}) follows from the fact that
$P_{\infty}(s)$ is uniformly bounded in $ \mathbb{R}\setminus
[-1-\delta,1+\delta]$ combined with the inequalities (\ref{gprops}).
It remains to settle (\ref{lipsbound}). On the contours $\Gamma_{\pm
}$, we have the explicit expression
%
%
\begin{eqnarray}
&& \Delta(s) = e^{\mp Nh(s)} %
\pmatrix{ P_{\infty}(s)_{12}P_{\infty}(s)_{22}
& - \bigl(P_{\infty
}(s)_{12} \bigr)^{2}\vspace*{3pt}
\cr
\bigl(P_{\infty}(s)_{22} \bigr)^{2} &
-P_{\infty
}(s)_{12}P_{\infty}(s)_{22} },
\nonumber
\\[-8pt]
\\[-8pt]
\eqntext{\displaystyle s \in\Gamma_{\pm},}
\end{eqnarray}
where $h(s)$ was defined in (\ref{hdef}). By Lemma~\ref{leszebound},
we see that $P_{\infty}(s)$ is uniformly bounded on $\Gamma_{\pm}$.
Therefore, the only danger is that $\operatorname{Re}{h(s)}$ vanishes
too quickly as
$N \to\infty$. However, a careful examination of the function (\ref
{hdef}) shows that $\operatorname{Re}{h(z)}$ vanishes at the same rate
that the
contours $\Gamma_{\pm}$ collapse onto the real axis. Indeed, an
elementary calculation using Taylor's theorem shows that we have the
inequalities
%
\begin{eqnarray}
\operatorname{Re}\bigl(h(s) \bigr) &>&  c_{1}/d_{N}, \qquad s
\in\Gamma_{+},
\nonumber
\\[-8pt]
\\[-8pt]
\nonumber
\operatorname{Re}\bigl(h(s) \bigr) &<& -c_{1}/d_{N}, \qquad s \in
\Gamma_{-},
\end{eqnarray}
where $c_{1}=4\eta\sqrt{1-(1-\delta)^{2}}$. This completes the proof
of (\ref{lipsbound}).
\end{pf}

\subsection*{Estimating the $R$-matrix and the proof of
Proposition~\texorpdfstring{\protect\ref{11111}}{}}
Finally, we are in a position to prove Proposition~\ref{11111}.
The proof follows from the standard method described in \cite
{DKMVZ99}. However, in our case extra care must be taken with the
estimates because our contour $\Sigma_{R}$ depends explicitly on $N$;
see, for example, \cite{BF06} for another example of
$N$-dependent contours.
%

\begin{proposition}
The matrix $R(z)$ satisfies the following estimate:
%
%
\begin{equation}
\label{Rasympt}\hspace*{6pt} R(z) = I + \mathcal{O} \biggl(\frac{1}{N} \biggr)+
\mathcal{O} \biggl(\log(d_{N})\exp\biggl(-c_{1}
\frac{N}{d_{N}} \biggr) \biggr), \qquad N \to\infty
\end{equation}
uniformly for $z \in\mathbb{C}\setminus\Sigma_{R}$.
\end{proposition}

\begin{pf}
Since for every $N$, $\Sigma_{R}$ is a finite union of smooth
contours, standard theory (see, e.g., \cite
{De99,KMVaV04,K07}) gives
%
%
\begin{equation}\label{rsol}
R(z) = I+\frac{1}{2\pi i}\int_{\Sigma_{R}}\frac{\Delta(s)\nu
(s)}{s-z}\,ds,
\end{equation}
where $\Delta(s)$ is as in Proposition~\ref{propdeltasigr} and $\nu
(s)$ is the unique solution to the singular integral equation $\nu(s)
= I+C_{-}[\nu\Delta](s)$. Here, $C_{-}$ is the Cauchy operator on
$L^{2}(\Sigma_{R})$, defined by
%
%
\begin{equation}
C_{-}[f](s) = \frac{1}{2\pi i}\int_{\Sigma_{R}}
\frac
{f(x)}{x-s_{-}}\,dx, \qquad f \in L^{2}(\Sigma_{R}),
\end{equation}
where $s_{-}$ denotes the limiting value of the integral as the point
$s \in\Sigma_{R}$ is approached from the minus side of the contour.

We begin by solving the equation for $\nu(s)$ in a perturbation series
(see, e.g., \cite{B08})
%
%
\begin{equation}
\label{pert} \nu(s)=I+\sum_{k=1}^{\infty}
\nu_{k}(s), \qquad\nu_{k}(s) = C_{-}[
\nu_{k-1}\Delta](s),
\end{equation}
and $\nu_{0}=I$. We need to show that this series is absolutely and
uniformly convergent for any $s \in\Sigma_{R}$.
%
%
\begin{figure}

\includegraphics{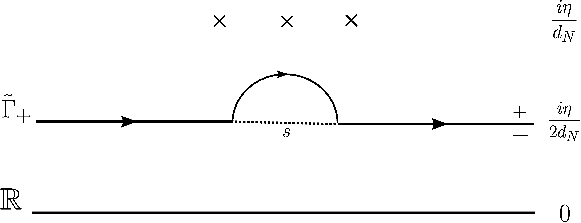}

\caption{The deformed contour $\tilde{\Gamma}_{+}$. The semi-circle
of radius $\eta/(4d_{N})$ is sufficiently small that it does not touch
the singularities (crosses), whose imaginary parts are $\eta/d_{N}$.}
\label{figdefcontour}
\end{figure}
%
Let $s \in\Gamma_{+}$ and deform $\Gamma_{+}$ to a new contour
$\tilde{\Gamma}_{+}$ differing only by a small semi-circle of radius
$\eta/(4d_{N})$ centered at $s$, as depicted in Figure~\ref
{figdefcontour}. Denote by $\tilde{\Sigma}_{R}$ the contour $\Sigma
_{R}$ with $\Gamma_{+}$ replaced with $\tilde{\Gamma}_{+}$. By the
Cauchy theorem, we have
%
%
\begin{equation}
\label{deform} \nu_{1}(s) = \frac{1}{2\pi i}\int
_{\Sigma_{R}}\frac{\Delta
(x)}{x-s_{-}}\,dx= \frac{1}{2\pi i}\int
_{\tilde{\Sigma}_{R}}\frac
{\Delta^{(0)}(x)}{x-s}\,dx,
\end{equation}
where $\Delta^{(0)}$ is the analytic continuation of $\Delta$ to
$\tilde{\Sigma}_{R}$ and satisfies the same bounds as in
Proposition~\ref{propdeltasigr}. Now we estimate, splitting the
integral into a
contribution from the discs $\partial B_{\pm1}(\delta)$, the real
line $\mathbb{R}\setminus[-1-\delta,1+\delta]$ [both of which are
at most $\mathcal{O}(N^{-1})$] and the contribution from $\tilde
{\Gamma}_{\pm}$:
%
%
\begin{eqnarray}
\bigl|\nu_{1}(s)\bigr| &\leq& c_{3}/N+
\frac{1}{2\pi}\int_{\tilde{\Gamma
}_{\pm}}\frac{|\Delta^{(0)}(x)|}{|x-s|}\,dx\nonumber
\\
\label{nuest}
&\leq& c_{3}/N + \frac{1}{2\pi}e^{-c_{1}N/d_{N}}\int
_{\tilde{\Gamma
}_{\pm}}\frac{1}{|x-s|}\,dx
\\
&\leq& c_{3}/N+c_{2}\log(d_{N})e^{-c_{1}N/d_{N}},
\qquad s \in\Gamma_{+},\nonumber
\end{eqnarray}
where $c_{3}$ and $c_{2}$ are constants depending only on $\delta$ and
$\eta$, with a similar bound if $s \in\Gamma_{-}$. If $s \in\Sigma
_{R}\setminus(\Gamma_{+}\cup\Gamma_{-})$, then the same bound holds
with $c_{2}=0$. Applying this procedure inductively, we obtain
%
%
\begin{eqnarray}\label{induct}
&& \bigl|\nu_{j}(s)\bigr| \leq K_{1}N^{-j}+K_{2}
\bigl(\log(d_{N})e^{-c_{1}N/d_{N}} \bigr)^{j}, \qquad s
\in\Sigma_{R},
\end{eqnarray}
where we can choose $K_{2}=0$ if $s \in\Sigma_{R}\setminus(\Gamma
_{+}\cup\Gamma_{-})$. The bound (\ref{induct}) implies that the
series (\ref{pert}) is absolutely convergent. Inserting (\ref{pert})
back into (\ref{rsol}), we arrive at
%
%
\begin{eqnarray}
R(z)  &=&  I+\sum_{j=1}^{\infty}R_{j}(z),
\nonumber
\\[-8pt]
\label{Rseries}
\\[-8pt]
\nonumber
 R_{j}(z) &=& \frac
{1}{2\pi i}\int_{\Sigma_{R}}
\frac{\nu_{j-1}(s)\Delta(s)}{s-z}\,ds, \qquad j=1,2,3,\ldots.
\end{eqnarray}
Now we bound the terms in the sum (\ref{Rseries}). First, consider the
case that $\operatorname{dist}(z,\Sigma_{R}) \geq\eta/(4d_{N})$. Then
estimates entirely analogous to (\ref{nuest}) yield
%
%
\begin{equation}
\label{rbound}\quad \bigl|R_{j}(z)\bigr| \leq K_{1}N^{-j}+K_{2}
\bigl(\log(d_{N})e^{-c_{2}N/d_{N}} \bigr)^{j}, \qquad
j=1,2,3,\ldots.
\end{equation}
On the other hand, if $0<\operatorname{dist}(z,\Sigma_{R}) < \eta
/(4d_{N})$, one can again deform the contour with a semi-circle of
radius $\eta/(4d_{N})$ and obtain the same bound (\ref{rbound}) after
essentially repeating the steps (\ref{deform}) and (\ref{nuest}).
\end{pf}


\begin{remark}
To complete the proof of Proposition~\ref{11111}, we will derive
the explicit form of the $O(1/N)$ term in (\ref{Rasympt}). Thus, we
need to compute the function $R_{1}(z)$ defined in (\ref{Rseries}). By
Proposition~\ref{propdeltasigr} and Lemma~\ref{lediscestlem}, we have
%
%
\begin{equation}\label{r1est}
R_{1}(z) = \frac{\tilde{R}_{1}(z)}{N}+\mathcal{O} \biggl(\frac
{1}{Nd_{N}}
\biggr)+\mathcal{O} \biggl(d_{N}\exp\biggl(-c_{1}
\frac{N}{d_{N}} \biggr) \biggr),
\end{equation}
where
%
%
\begin{equation}
\label{r1tilde} \tilde{R}_{1}(z) = \frac{1}{2\pi i}\int
_{\partial B_{1}(\delta
)}\frac{\Delta^{(+1)}_{1}(s)}{s-z}\,ds+\frac{1}{2\pi i}\int
_{\partial B_{-1}(\delta)}\frac{\Delta^{(-1)}_{1}(s)}{s-z}\,ds.
\end{equation}
The functions $\Delta^{(\pm1)}_{1}(s)$ are explicitly known, for example, by setting $\omega(z)\equiv1$ in equations (79), (83),
of \cite{K07} or by using the results in \cite{DKMVZ99}. Then
expanding (\ref{r1est}) near $z=\infty$ and computing the residues of
the function $\Delta^{(\pm1}_{1}(s)$ near the poles $s= \pm1$, we
find that
%
%
\begin{equation}
\tilde{R}_{1}(z) = A/z+B/z^{2}+O \bigl(z^{-3}
\bigr), \qquad z \to\infty,
\end{equation}
where
%
%
\begin{equation}
\label{constantmatrix} A = %
\pmatrix{ 0 & i/24\vspace*{2pt}
\cr
i/24 & 0 },
\qquad B = \pmatrix{ -1/48 & 0\vspace*{2pt}
\cr
0 & 1/48}.
\end{equation}
Then inserting (\ref{Rasympt}) and the first-order correction above
into the definition (\ref{rsp}), we arrive at (\ref{zkas}).
\end{remark}

\section{The Szeg\"o function}
For a weight $\omega(x)$, the Szeg\"o function is defined by the formula
%
%
\begin{equation}
\mathcal{D}(z) = \exp\biggl(\frac{\sqrt{z+1}\sqrt{z-1}}{2\pi
}\int_{-1}^{1}
\frac{\log(\omega(x))}{\sqrt{1-x^{2}}}\frac
{dx}{z-x} \biggr).
\end{equation}
It satisfies the properties:
\begin{longlist}[3.]
\item[1.] $\mathcal{D}(z)$ is nonzero and analytic in $\mathbb
{C}\setminus[-1,1]$,
\item[2.] $\mathcal{D}_{+}(x)\mathcal{D}_{-}(x) = \omega(x)$ for $x \in
(-1,1)$,
\item[3.] $\lim_{z \to\infty}\mathcal{D}(z) = \mathcal{D}_{\infty}
\neq0$.
\end{longlist}
For our problem, we are interested in the weight $\omega(x) = \prod
_{k=1}^{m}|x-z_{k}|^{\alpha_{k}}$ where $\operatorname{Im}(z_{k})
\neq0$ for
$k=1,\ldots,m$. It can easily be seen that the above three properties
uniquely specify the Szeg\"o function for this weight. Let $c(z) =
z+\sqrt{z-1}\sqrt{z+1}$ be the conformal map from $\mathbb
{C}\setminus[-1,1]$ to the exterior of the unit disk. Then the Szeg\"o
function for the weight $|x-\mu|^{2}$ is
%
%
\begin{equation}
\frac{|c(\mu)|}{2} \biggl(1-\frac{1}{c(\mu)c(z)} \biggr) \biggl(1-
\frac{1}{\overline{c(\mu)}c(z)} \biggr), \qquad\operatorname
{Im}(\mu) \neq0.
\end{equation}
This can be checked by verifying the above three conditions using the
properties $c(z)+\frac{1}{c(z)}=2z$ and $c_{+}(x)c_{-}(x) = 1$ for $x
\in[-1,1]$. Thus, the Szeg\"o function for $\omega(x)$ is
%
%
\begin{equation}
\label{exactszego} \mathcal{D}(z) = \prod_{k=1}^{m}
\biggl(\frac{|c(z_{k})|}{2} \biggl(1-\frac{1}{c(z_{k})c(z)} \biggr
) \biggl(1-
\frac{1}{\overline
{c(z_{k})}c(z)} \biggr) \biggr)^{\alpha_{k}/2}.
\end{equation}
Similar considerations show straightforwardly that the function
$C(z,\mu)$ defined in~(\ref{phifn}) is given by
%
%
\begin{equation}
C(z,\mu) = \frac{1}{2}\log\biggl(\frac{|c(\mu)|}{2} \biggl(1-
\frac
{1}{c(\mu)c(z)} \biggr) \biggl(1-\frac{1}{\overline{c(\mu
)}c(z)} \biggr) \biggr).
\end{equation}
Defining $z_{k}= x_{0}+\frac{\tau_{k}+i\eta}{d_{N}}$, one easily
gets the asymptotic
%
%
\begin{eqnarray}
&& d_{N}\frac{|c(z_{j})|}{2} \biggl(1-\frac{1}{c(z_{j})c(z_{k})} \biggr)
\biggl(1-\frac{1}{\overline{c(z_{j})}c(z_{k})} \biggr)
\nonumber
\\[-8pt]
\\[-8pt]
\nonumber
&&\qquad= 2\eta
+i(\tau_{j}-
\tau_{k}) + \mathcal{O} \bigl(d_{N}^{-1} \bigr)
\end{eqnarray}
which immediately implies that
%
%
\begin{equation}
\label{szasy2} \hspace*{5pt}\operatorname{Re}\bigl(C(z_{j},z_{k}) \bigr) = -
\tfrac{1}{2}\log(d_{N})+\tfrac{1}{4}\log\bigl((
\tau_{j}-\tau_{k})^{2}+4\eta^{2}
\bigr)+\mathcal{O} \bigl(d_{N}^{-1} \bigr).
\end{equation}
The uniformity of the error term in the relevant compact sets follows
from the uniform expansions of the logarithm and square roots in these
regions. From (\ref{exactszego}), we obviously have the expansion
%
%
\begin{equation}
\mathcal{D}(z) = \mathcal{D}_{\infty} \biggl(1+\frac{\mathcal
{D}_{1}}{z}+
\frac{\mathcal{D}_{1}^{2}/2+\mathcal
{D}_{2}}{z^{2}} \biggr)+O \bigl(z^{-3} \bigr),
\end{equation}
where
%
%
\begin{equation}\label{formdinf}
\mathcal{D}_{\infty} = \prod_{k=1}^{m-1}
\biggl|\frac
{c(z_{k})}{c(z_{m})} \biggr|^{\alpha_{k}/2}
\end{equation}
and
%
%
\begin{equation}
\label{higherord}
\mathcal{D}_{1} = -\frac{1}{2}\sum
_{k=1}^{m}\alpha_{k}\operatorname{Re}\biggl(
\frac{1}{c(z_{k})} \biggr), \qquad\mathcal{D}_{2} = -\frac
{1}{8}
\sum_{k=1}^{m}\alpha_{k}\operatorname{Re}
\biggl(\frac
{1}{c(z_{k})^{2}} \biggr).
\end{equation}

\section{Proof of equation \texorpdfstring{(\protect\ref{432432})}{(4.29)}}
\label{apintcomp}
Our first task is to prove that we have the limit
%
%
\begin{equation}
\label{zerolim} \lim_{N \to\infty}\int_{[I_{N}^{c}]^{2}}
\frac{\Delta
f_{1}(d_{N}x)}{\Delta x}\frac{\Delta f_{2}(d_{N} x)}{\Delta
x}F_{N}(x_{1},x_{2})
\,dx_{1}\,dx_{2}=0,
\end{equation}
where $I_{N}^{c}$ is the complement of the region $I_{N} = [-(1-\delta
_{N}),(1-\delta_{N})]$, $\delta_{N} = N^{-7/12}$ and we defined
$F_{N}(x,y)=(x-y)^{2}K_{N}^{2}(x,y)$ in terms of the GUE kernel (\ref
{guekernel}). After proving (\ref{zerolim}), we show that $\delta_{N}$
can be replaced with an $N$-independent $\delta>0$ costing an error
term that can be neglected.

Let $0<\varepsilon<1$ and consider the following three subsets of
$\mathbb{R}^{2}$:
\begin{eqnarray}
R_{1} &=& \bigl\{(x_{1},x_{2}) \in
\mathbb{R}^{2} \vert\bigl(|x_{1}| < \varepsilon\bigr) \land
\bigl(x_{2}>(1+\delta_{N}) \bigr) \bigr\},
\nonumber
\\
R_{2} &=& \bigl\{(x_{1},x_{2}) \in
\mathbb{R}^{2} \vert\bigl(|x_{1}| < \varepsilon\bigr) \land(1-
\delta_{N} < x_{2} < 1+\delta_{N}) \bigr\},
\nonumber
\\
R_{3} &=& \bigl\{(x_{1},x_{2}) \in
\mathbb{R}^{2} \vert(x_{1}>\varepsilon) \land(x_{2}>
\varepsilon) \bigr\}.
\nonumber
\end{eqnarray}
It is sufficient\vspace*{1pt} to consider only these regions, because together with
their reflections in the $x_{1}$ and $x_{2}$ axes, they cover the
entire region $[I_{N}^{c}]^{2}$. In the following, we will prove that
the contribution from each of these regions to the integral (\ref
{zerolim}) tends to zero as $N \to\infty$. Finally, we complete the
proof of equation (\ref{eqnfour}) by showing that the difference
between the integral (\ref{zerolim}) over $[I_{N}^{c}]^{2}$ and
$[I_{\delta}^{c}]^{2}$ converges as $N \to\infty$ to a function that
is $O(\delta)$ as $\delta\to0$.

We start with the contribution of the region $R_{3}$ to the integral
(\ref{zerolim}). Using the Schwartz property of $f_{1}, f_{2}$ and the
inequality $K_{N}^{2}(x_{1},x_{2}) \break \leq N^{2}\rho_{N}(x_{1})\rho
_{N}(x_{2})$, we have for any $\gamma>0$
%
%
\begin{eqnarray}
&& \Biggl|\int_{\varepsilon}^{\infty}\!\!\int_{\varepsilon}^{\infty}
\Delta f_{1}(d_{N} x)\Delta f_{2}(d_{N}
x) K_{N}^{2}(x_{1},x_{2})
\,dx_{1}\, dx_{2} \Biggr|
\\
&&\qquad \leq N^{2}(2\varepsilon d_{N})^{-2\gamma} \biggl(\int
_{\varepsilon
}^{\infty}\rho_{N}(x_{1})
\,dx_{1} \biggr) \biggl(\int_{\varepsilon
}^{\infty}
\rho_{N}(x_{2})\,dx_{2} \biggr)
\nonumber
\\[-8pt]
\\[-8pt]
\nonumber
&&\qquad= O
\bigl(N^{2}d_{N}^{-\infty} \bigr),
\end{eqnarray}
where we used the inequality $|\Delta g_{j}(d_{N}x)| \leq|g_{j}(d_{N}
x_{1})+g_{j}(d_{N} x_{2})| \leq\break  d_{N}^{-\gamma}(|x_{1}|^{-\gamma
}+|x_{2}|^{-\gamma}) \leq 2d_{N}^{-\gamma}(\varepsilon^{-\gamma})$. We\vspace*{1pt}
conclude that the integral (\ref{orig}) restricted to the region
$R_{3}$ is of order $O(N^{-\infty})$ as $N \to\infty$.

Now let us consider the edge region $R_{2}$. We will make use of the
following lemma from \cite{PS11}, which states
%

\begin{lemma}[(Theorem~5.2.3(ii) \cite{PS11})]
\label{leexpdecaylem}
Let $\rho_{N}(x)$ denote the normalized density of states, as in
(\ref{dos}). The bound
%
%
\begin{equation}
\rho_{N} \bigl(1+sN^{-2/3} \bigr) \leq\bigl(BN^{1/3}s
\bigr)^{-1}e^{-bs^{3/2}}
\end{equation}
holds for $N$ large enough. Here, $B$ and $b$ are absolute constants
and $s \to\infty$ as $N \to\infty$.
\end{lemma}

Using this result and again the bound $K_{N}(x_{1},x_{2})^{2} \leq
N^{2}\rho_{N}(x_{1})\rho_{N}(x_{2})$, we see that the contribution to
the integral (\ref{zerolim}) from the region $R_{2}$ is bounded by
%
%
\begin{eqnarray}
&& N^{2}\int_{-\infty}^{\infty}\int
_{(1+\delta_{N})}^{\infty
}\bigl|\Delta f_{1}(d_{N}x)\bigr|\bigl|
\Delta f_{2}(d_{N}x)\bigr|\rho_{N}(x_{1})\rho
_{N}(x_{2})\,dx_{1}\,dx_{2}
\nonumber
\\
&&\qquad= C\delta_{N}N^{2}\int_{-\infty}^{\infty}
\int_{1}^{\infty}\rho_{N}(1+x_{1}
\delta_{N})\rho_{N}(x_{2})\,dx_{1}
\,dx_{2}
\nonumber
\\[-8pt]
\label{inequal}
\\[-8pt]
\nonumber
&&\qquad\leq  CBN\int_{-\infty}^{\infty}\int_{1}^{\infty
}x_{1}^{-1}e^{-bx_{1}^{3/2}N^{1/8}}
\rho_{N}(x_{2})\,dx_{1}\,dx_{2} = O
\bigl(N^{-\infty} \bigr), \\
\eqntext{\displaystyle N \to\infty,}
\end{eqnarray}
where we used that $f_{1}$, $f_{2}$ are uniformly bounded on $\mathbb{R}^{2}$.

For the region $R_{1}$, we need a bound for the absolute value of the
functions $\psi_{l}^{(N)}(x)$.
%

\begin{lemma}[(Szeg\"o, Section~10.8 \cite{Sze39})]
Let $\psi^{(N)}_{l}(x)$ denote the orthonormal functions defined in
(\ref{wavefns}). Then the following bound holds uniformly in $l$ as \mbox{$N
\to\infty$}:
%
%
\begin{equation}
\sup_{u \in\mathbb{R}}\bigl|\psi^{(N)}_{l}(u)\bigr| = O
\bigl(N^{1/4} \bigr).
\end{equation}
\end{lemma}

First, consider the contribution from the product of squares, that is,
that of $\psi^{(N)}_{N}(x_{1})^{2}\psi^{(N)}_{N-1}(x_{2})^{2}$ in
$F_{N}(x_{1},x_{2})$. Since in the region $R_{1}$ we have $x_{1} \neq
x_{2}$, the bound $|\Delta f_{j}(d_{N}x)/ \Delta x| \leq C$, $j=1,2$
holds for some $N$-independent $C>0$. Then the contribution coming from
$\psi^{(N)}_{N}(x_{1})^{2}\psi^{(N)}_{N-1}(x_{2})^{2}$ is bounded by
%
%
\begin{eqnarray}
&&\hspace*{3pt}\quad C\int_{(1-\delta_{N})}^{(1+\delta_{N})}\!\int_{-\varepsilon}^{\varepsilon
}
\psi^{(N)}_{N}(x_{1})^{2}
\psi^{(N)}_{N-1}(x_{2})^{2}
\,dx_{1}\, dx_{2}
\\
&&\hspace*{3pt}\quad\qquad \leq C\int_{(1-\delta_{N})}^{(1+\delta_{N})}\!\int_{-\infty
}^{\infty}
\psi^{(N)}_{N}(x_{1})^{2}\sup
_{u \in\mathbb{R}}\bigl|\psi^{(N)}_{N-1}(u)\bigr|^{2}
\,dx_{1}\,dx_{2} \leq C'N^{-1/12},
\end{eqnarray}
where $C'>0$ is another constant independent of $N$. A similar
calculation shows that the contribution from the mixed term $\psi
^{(N)}_{N}(x_{1})\psi^{(N)}_{N-1}(x_{1})\times\break \psi^{(N)}_{N}(x_{2})\psi
^{(N)}_{N-1}(x_{2})$ is also $O(N^{-1/12})$ as $N \to\infty$. We
conclude that the contribution of the region $R_{1}$ is $O(N^{-1/12})$
as $N \to\infty$. Finally, a completely analogous calculation shows
that the contribution to (\ref{orig}) coming from all reflections of
the regions $R_{1}$, $R_{2}$ and $R_{3}$ in the $x_{1}$ and $x_{2}$
axes satisfy the same corresponding asymptotic estimates as $N \to
\infty$ and, therefore, may be neglected. Equation (\ref{zerolim}) is proven.

To complete the argument, we need to show that the difference between
the integral (\ref{orig}) over $I_{N}^{2}$ and the same integral over
$I_{\delta}=[-(1-\delta),(1-\delta)]^{2}$ for some $N$-independent
$\delta>0$, can be neglected in the limit $N \to\infty$. It will be
sufficient to consider only the thin strip $|x_{1}|<\varepsilon$ and
$(1-\delta) < x_{2} < (1-\delta_{N})$, because the remaining parts of
$I_{N}^{c}\setminus I_{\delta}$ are either reflections of this region
or are subsets of the region $R_{1}$ treated earlier. Thus, we just
have to estimate the integral
%
%
\begin{equation}
\label{strip} \int_{(1-\delta)}^{(1-\delta_{N})}\!\int
_{-\varepsilon
}^{\varepsilon}\frac{\Delta f_{1}(d_{N}x)}{\Delta x}\frac{\Delta
f_{2}(d_{N}x)}{\Delta x}F_{N}(x_{1},x_{2})
\,dx_{1}\,dx_{2}.
\end{equation}
According to the first Plancherel--Rotach formula of Corollary~5.1.5 in
\cite{PS11}, we have the bound $F_{N}(x_{1},x_{2}) =
(1-x_{1}^{2})^{-1/2}(4-x_{2}^{2})^{-1/2}O(1)$ uniformly as \mbox{$N \to
\infty$.} Therefore, since $x_{1} \neq x_{2}$ in (\ref{strip}) and
$f_{1},f_{2}$ are uniformly bounded, we see that (\ref{strip}) is
bounded in absolute value by
%
%
\begin{eqnarray}
&& C \biggl|\int_{(1-\delta)}^{(1-\delta_{N})}\!\int_{-\varepsilon}^{\varepsilon}
\bigl(1-x_{1}^{2} \bigr)^{-1/2}
\bigl(1-x_{2}^{2} \bigr)^{-1/2}\, dx_{1}
\,dx_{2} \biggr|
\\
&&\qquad \leq C\bigl| \bigl(\cos^{-1}(1-\delta_{N})-\cos^{-1}(1-
\delta) \bigr)\bigr| \to C\bigl|\cos^{-1}(1-\delta)\bigr|,
\nonumber
\\[-8pt]
\\[-8pt]
\eqntext{\displaystyle N \to\infty,}
\end{eqnarray}
where $C>0$ is some $N$-independent constant. Hence, by choosing
$\delta>0$ sufficiently small, we can ensure that the integral over
this strip is as small as we desire. This proves equation (\ref{432432}).
\end{appendix}


\section*{Acknowledgements}
We would like to thank Paul Bourgade, Arno Kuijlaars and Leonid Pastur
for insightful discussions and correspondence relating to our results,
Philippe Sosoe and Percy Wong for sharing their preprint \cite{SW13}
with us and J\'er\'emie Unterberger for bringing the paper \cite{U09}
to our attention. We are particularly grateful to Igor Krasovsky for
informative discussions about the Riemann--Hilbert problem, and also to
anonymous referees for their helpful suggestions and constructive
critique of the first version of the paper. The second author thanks
the Isaac Newton Institute, Cambridge, UK, for its support and
hospitality during the semester Periodic and Ergodic Spectral Problems.

%





\printaddresses
\end{document}